\documentclass[a4paper]{llncs}
\usepackage[utf8]{inputenc}
\usepackage{amssymb}
\usepackage{amsmath,bm}
\usepackage{mathrsfs}
\usepackage{mathrsfs}
\usepackage{xcolor,xspace}
\usepackage{setspace}
\usepackage{scalerel}
\usepackage{theorem}
\usepackage{cite}
\usepackage{url}
\usepackage{orcidlink}
\usepackage{array,multirow}
\usepackage{pgfplots}
\usetikzlibrary{plotmarks,patterns}
\usepackage[ruled,vlined,titlenumbered,linesnumbered]{algorithm2e}
\usepackage[nolist]{acronym}

\usepackage{makecell}
\usepackage{mleftright}

\usepackage{mathtools}
\mathtoolsset{showonlyrefs}

\pgfplotsset{
compat=1.17,
mystyle/.style={
    scale only axis,
    width=0.7\columnwidth,
    height=0.5\columnwidth,
    label style={inner sep=0, font=\normalsize}, 
    tick label style={font=\scriptsize},
    legend style={font=\scriptsize},
    mark size=3,
    major grid style={dashed},
    line width=0.8pt,
    axis line style = thin}
}

\newcolumntype{M}[1]{>{\centering\arraybackslash}m{#1}}

\newcommand{\Fqm}{\ensuremath{\mathbb F_{q^m}}}
\newcommand{\Fqms}{\ensuremath{\mathbb F_{q^{ms}}}}

\newcommand{\Fq}{\ensuremath{\mathbb F_{q}}}

\newcommand{\F}{\ensuremath{\mathbb F}}

\newcommand{\NN}{\ensuremath{\mathbb{N}}}

\newcommand{\dminsr}{\ensuremath{d}}

\SetKwComment{Comment}{}{}

\newcommand{\ext}{\ensuremath{\text{ext}}}

\newcommand{\defeq}{:=}

\DeclareMathOperator{\wt}{wt}
\DeclareMathOperator{\rk}{rk}
\DeclareMathOperator{\rkq}{rk_{q}}
\DeclareMathOperator{\rkqm}{rk_{q^m}}

\DeclareMathOperator{\supp}{supp}

\DeclareMathOperator{\diag}{diag}
\DeclareMathOperator{\REF}{REF}

\newcommand{\rker}{\ensuremath{\ker_r}}

\renewcommand{\vec}[1]{\ensuremath{\bm{#1}}}
\newcommand{\mat}[1]{\ensuremath{\bm{#1}}}

\renewcommand{\b}{\vec{b}}
\renewcommand{\c}{\vec{c}}

\newcommand{\g}{\vec{g}}

\newcommand{\n}{\vec{n}}

\renewcommand{\v}{\vec{v}}

\newcommand{\x}{\vec{x}}
\newcommand{\y}{\vec{y}}

\newcommand{\A}{\mat{A}}
\newcommand{\B}{\mat{B}}
\newcommand{\C}{\mat{C}}
\newcommand{\D}{\mat{D}}
\newcommand{\E}{\mat{E}}
\newcommand{\G}{\mat{G}}
\newcommand{\I}{\mat{I}}
\newcommand{\J}{\mat{J}}
\renewcommand{\H}{\mat{H}}

\renewcommand{\P}{\mat{P}}

\renewcommand{\S}{\mat{S}}

\newcommand{\X}{\mat{X}}
\newcommand{\Y}{\mat{Y}}

\newcommand{\0}{\ensuremath{\mathbf 0}}

\newcommand{\mycode}[1]{\ensuremath{\mathcal{#1}}}

\newcommand{\SumRankWeight}{\ensuremath{\wt_{\Sigma R}}}

\newcommand{\SkewWeight}{\ensuremath{\wt_{skew}}}

\newcommand{\SumRankDist}{d_{\ensuremath{\Sigma}R}}

\newcommand{\RankSupp}{\supp_{R}}
\newcommand{\SumRankSupp}{\supp_{\Sigma R}}
\newcommand{\HammingSupp}{\supp_{H}}
\newcommand{\myspace}[1]{\mathcal{#1}}
\newcommand{\Rowspace}[1]{\ensuremath{\myspace{R}_{q}\!\left(#1\right)}}
\newcommand{\RowspaceFqm}[1]{\ensuremath{\myspace{R}_{q^m}\!\left(#1\right)}}
\newcommand{\oh}[1]{\bnd{O}{#1}}

\newcommand{\bnd}[2]{\ensuremath{#1\mathopen{}\left(#2\right)\mathclose{}}}

\newcommand{\intOrder}{\ensuremath{s}}

\newcommand{\shots}{\ensuremath{\ell}}

\newcommand{\todo}[1]{}
\newcommand{\tj}[1]{}
\newcommand{\hb}[1]{}
\newcommand{\fh}[1]{}
\newcommand{\new}[1]{}

\newcommand{\MeKap}{Metzner--Kapturowski}

\DeclareMathOperator{\mksum}{\ensuremath{\bigoplus}}

\newcommand{\h}{\vec{h}}

\newcommand{\Hsub}{\ensuremath{\H_{\text{sub}}}}

\begin{acronym}
 \acro{BCH}{Bose--Chaudhuri--Hocquenghem}
 \acro{BMD}{bounded minimum distance}
 \acro{ESP}{error span polynomial}
 \acro{ELP}{error locator polynomial} 
 \acro{SRS}{skew Reed--Solomon}
 \acro{ISRS}{interleaved skew Reed--Solomon}
 \acro{lclm}{least common left multiple}
 \acro{LRS}{linearized Reed--Solomon}
 \acro{LLRS}{lifted linearized Reed--Solomon}
 \acro{ILRS}{interleaved linearized Reed--Solomon}
 \acro{LILRS}{lifted interleaved linearized Reed--Solomon}
 \acro{MDS}{maximum distance separable}
 \acro{MSRD}{maximum sum-rank distance}
 \acro{MSD}{maximum skew distance}
 \acro{RS}{Reed--Solomon}
 \acro{REF}{row echelon form}
 \acro{LEEA}{linearized extended Euclidean algorithm}
 \acro{SEEA}{skew extended Euclidean algorithm}
 \acro{VSD}{vector-symbol decoding}
 \acro{AG}{algebraic-geometry}
 \acro{KEM}{key encapsulation mechanism}
 \acro{ISD}{information-set-decoding}
\end{acronym}

\begin{document}

    \title{On Decoding High-Order Interleaved Sum-Rank-Metric Codes\vspace*{-.5cm}}

    \author{Thomas Jerkovits$^{1,3}$ \orcidlink{0000-0002-7538-7639} \and Felicitas Hörmann$^{1,2}$ \orcidlink{0000-0003-2217-9753} \and Hannes Bartz$^1$ \orcidlink{0000-0001-7767-1513}}
    \institute{
        $^1$~Institute of Communications and Navigation,\\German Aerospace Center (DLR), Germany \\
        $^2$~School of Computer Science, University of St.\ Gallen, Switzerland \\
        $^3$~Institute for Communications Engineering, Technical University of Munich (TUM), Germany
        \email{$\{$thomas.jerkovits, felicitas.hoermann, hannes.bartz$\}$@dlr.de}}

    \maketitle

\begin{abstract}
 We consider decoding of vertically homogeneous interleaved sum-rank-metric codes with high interleaving order $\intOrder$, that are constructed by stacking $\intOrder$ codewords of a single constituent code.
 We propose a \MeKap-like decoding algorithm that can correct errors of sum-rank weight $t \leq d-2$, where $d$ is the minimum distance of the code, if the interleaving order $\intOrder \geq t$ and the error matrix fulfills a certain rank condition.
 The proposed decoding algorithm generalizes the \MeKap(-like) decoders in the Hamming metric and the rank metric and has a computational complexity of $\oh{\max\{n^3, n^2\intOrder\}}$ operations in $\Fqm$
% \fh{Can we make clear that this is in terms of $\Fqm$-operations?}
 , where $n$ is the length of the code.
 The scheme performs linear-algebraic operations only and thus works for
% \fh{"the interleaving of"? We still need this structure.}
 any
% vertically homogenous
 interleaved linear sum-rank-metric code.
 We show how the decoder can be used to decode high-order interleaved codes in the skew metric.
 Apart from error control, the proposed decoder allows to determine the security level of code-based cryptosystems based on interleaved sum-rank metric codes.
\end{abstract}

%\fh{Should we mention that we consider \textit{homogeneous} interleaving somewhere in the abstract or the introduction? The first time this appears is in Section 2.2 where we define interleaved codes and describe the channel model.}
%    \tj{Add vertically also with explanation}

%----------------------------------------------------------------------------
% Introduction
%----------------------------------------------------------------------------
\section{Introduction} \label{sec:introduction}

The development of quantum-secure cryptosystems is crucial in view of the recent advances in the design and the
realization of quantum computers.
As it is reflected in the number of submissions during the NIST’s post-quantum cryptography standardization process for \acp{KEM}, many promising candidates belong to the family of code-based systems of which still three candidates are in the current 4th round~\cite{NISTreport2022}.
Code-based cryptography is mostly based on the McEliece cryptosystem~\cite{McEliece1978} whose trapdoor is that the
public code can only be efficiently decoded if the secret key is known.

Variants of the McEliece cryptosystem based on interleaved codes in the Hamming and the rank metric were proposed
in~\cite{elleuch2018public,holzbaur2019decoding,renner2019interleavingloidreau}.
Interleaving is a well-known technique in coding theory that enhances a code's burst-error-correction capability.
The idea is to stack a fixed number $s$ of codewords of a constituent code over a field $\Fqm$ in a matrix and thus to
transform burst errors into errors occurring at the same position in each codeword.
Equivalently, these errors can be seen as symbol errors in a vector code over the extension field $\Fqms$.
There exist list and/or probabilistic unique decoders for interleaved \ac{RS} codes
%\fh{Put the parentheses after "in the Hamming metric" to be consistent with the rest of the sentence.}
in the Hamming metric~\cite{krachkovsky1997decoding} as well as for interleaved Gabidulin codes in the rank metric~\cite{Loidreau_Overbeck_Interleaved_2006} and for interleaved \ac{LRS} codes in the sum-rank metric~\cite{bartz2022fast}
%\fh{Maybe use the same parentheses construction here to be more consistent. Human brains like parallel structures in sentences.}
%\tj{Maybe get rid of all the parentheses citation constructions? We are using this throughout the paper in arbitrary places. Either this follows a certain rule, which I'm not seeing right now or in my opinion we should get rid of it.}
%\fh{Also fine for me. :-)}
.

All of the mentioned decoders are tailored to a particular code family and explicitly exploit the code structure.
In contrast, Metzner and Kapturowski proposed a decoder which
%\fh{Grammatically, "which" refers to "codes" and not to "decoder", unfortunately.}
%\tj{Is this version better?}
%\fh{Yes, excellent!}
works for interleaved Hamming-metric codes with \emph{any}
linear constituent code.
The decoding algorithm only requires a high interleaving order $s$ as well as a linear-independence constraint
on the error~\cite{metzner1990general}.
Variants of the linear-algebraic Metzner--Kapturowski algorithm were further studied
in~\cite{oh1994performance, haslach1999decoding, haslach2000efficient, metzner2002vector, metzner2003vector, roth2014coding},
often under the name \ac{VSD}.
Moreover, Puchinger, Renner and Wachter-Zeh adapted the algorithm to the rank-metric case in~\cite{puchinger2019decoding,renner2021decoding}.

This affects the security level of McEliece variants that are based on interleaved codes in the Hamming and the rank
metric as soon as the interleaving order $s$ is too large (i.e. $s \geq t$ for error weight $t$).
Cryptosystems based on interleaved codes with small interleaving order are not affected.
Their security level can be evaluated based on \ac{ISD} algorithms (see e.g.~\cite{porwal2022interleavedPrange} for an
adaptation of Prange's algorithm to interleaved Hamming-metric codes).

\paragraph{Contribution}
We present a \MeKap-like decoding algorithm for high-order interleaved sum-rank-metric codes with an arbitrary linear
constituent code.
This gives valuable insights for the design of McEliece-like cryptosystems based on interleaved codes in the sum-rank
metric.
The proposed algorithm is purely linear-algebraic and can guarantee to correct errors of sum-rank weight $t \leq d-2$ if the error matrix has full $\Fqm$-rank $t$, where $d$ is the minimum distance of the code.
The computational complexity of the algorithm is in the order of $\oh{\max\{n^3, n^2\intOrder\}}$ operations over $\Fqm$, where $\intOrder \geq t$
is the interleaving order and $n$ denotes the length of the linear constituent code.
%\tj{Comment on decoding complexity of constituent code. Say that its generic.}
%\tj{What about: Note that no structure on the underlying constituent code is assumed for the derivation of the decoding complexity.}
%\fh{I like it. Only one minor thing: I would rather say "the derivation of the decoder and its complexity" but it's a matter of taste.}
Note, that the decoding complexity is independent of the code structure of the constituent code since the proposed algorithm exploits properties of high-order interleaving only.
Since the sum-rank metric generalizes both the Hamming and the rank metric, the original Metzner--Kapturowski
decoder~\cite{metzner1990general} as well as its rank-metric analog~\cite{puchinger2019decoding,renner2021decoding}
can be recovered from our proposal.

% \fh{Maybe remove this:
% \paragraph{Outline}
% We start by giving some preliminaries and the channel model in Section~\ref{sec:preliminaries}.
% Section~\ref{sec:decoder} is the main part of the paper and contains a detailed description of the \MeKap-like decoder
% for the sum-rank-metric case.
% We highlight how the special cases for the Hamming and the rank metric can be recovered in Section~\ref{sec:comparison}
% and conclude the paper with a summary and some open research problems in Section~\ref{sec:conclusion}.
% }

%----------------------------------------------------------------------------
% Preliminaries
%----------------------------------------------------------------------------
\section{Preliminaries} \label{sec:preliminaries}

%\begin{itemize}
%\item expansion map $\Fqm \to \Fq^m$
%\item notation for submatrices
%\item sum-rank and skew metric (and their isometry)
%\item homogeneous vertically interleaved codes
%\end{itemize}

% \subsection{Notation}

Let $q$ be a power of a prime and let $\Fq$ denote the finite field of order $q$ and $\Fqm$ an extension field of degree $m$.
We use $\Fq^{a\times b}$ to denote the set of all $a\times b$ matrices over $\Fq$ and $\Fqm^b$ for the set of all row vectors of length $b$ over $\Fqm$.

Let $\vec{b} = (b_1,\ldots,b_m) \in \Fqm^m$ be a fixed (ordered) basis of $\Fqm$ over $\Fq$.
We denote by $\ext(\alpha)$ the column-wise expansion of an element $\alpha \in \Fqm$ over $\Fq$ (with respect to $\b$), i.e.
\begin{equation*}
    \ext : \Fqm \mapsto \Fq^{m}
\end{equation*}
such that $\alpha = \vec{b} \cdot \ext(\alpha)$.
This notation is extended to vectors and matrices by applying $\ext(\cdot)$ in an element-wise manner.

By $[a:b]$ we denote the set of integers $[a:b]\defeq\{i:a\leq i\leq b\}$.
For a matrix $\A$ of size $a \times b$ and entries $A_{i,j}$ for $i \in [1:a]$  and $j \in [1:b]$, we define the submatrix notation
\begin{equation*}
    \A_{[c:d],[e:f]}\defeq
    \begin{pmatrix}
        A_{c,e} & \dots & A_{c,f}
        \\
        \vdots & \ddots & \vdots
        \\
        A_{d,e} & \dots & A_{d,f}
    \end{pmatrix}.
\end{equation*}

The $\Fqm$-linear row space of a matrix $\A$ over $\Fqm$ is denoted by $\RowspaceFqm{\A}$.
Its $\Fq$-linear row space is defined as $\Rowspace{\A} \defeq \myspace{R}_q{(\ext{(\A)})}$.
We denote the row-echelon form and the (right) kernel of $\A$ as $\REF(\A)$ and $\rker(\A)$, respectively.

\subsection{Sum-Rank-Metric Codes}
Let $\n=(n_1,\dots, n_\shots)\in\NN^\shots$ with $n_i > 0$ for all $i \in [1:\shots]$ be a length partition\footnote{Note that this is also known as (integer) composition into exactly $\shots$ parts in combinatorics.
%\tj{I also checked again. This is actually called a weak composition, since we also allow zeros.}\fh{Fun Fact: I checked a combinatorics book that I have at home again -- they define compositions with nonnegative values and weak ones with positive values, i.e. exactly the opposite from Wikipedia etc. :D We can keep it with weak because that's what is used in the first Google results. :-)} \tj{Haha great if literature is not consistent. Or maybe just write "is related to integer compositions known in combinatoris..."? :D}
}  of $n$, i.e. $n = \sum_{i=1}^{\shots} n_i$.
Further let $\x=(\x^{(1)} \,|\, \x^{(2)} \,|\, \dots \,|\, \x^{(\shots)})\in\Fqm^n$ be a vector over a finite field $\Fqm$ with $\x^{(i)}\in\Fqm^{n_i}$.
For each $\x^{(i)}$ define the rank $\rkq(\x^{(i)}) \defeq \rkq(\ext(\x^{(i)}))$ where $\ext(\x^{(i)})$ is a matrix in $\Fq^{m \times n_i}$ for all $i\in[1:\shots]$.
The \emph{sum-rank weight} of $\x$ with respect to the length partition $\n$ is defined as
\begin{equation*}
    \SumRankWeight(\x) \defeq \sum_{i=1}^{\shots}\rk_q(\x^{(i)})
\end{equation*}
and the \emph{sum-rank distance} of two vectors $\x, \y \in \Fqm^n$ is defined as $\SumRankDist(\x,\y) \defeq \SumRankWeight(\x-\y)$.
Note that the sum-rank metric equals the Hamming metric for $\shots=n$ and is equal to the rank metric for $\shots=1$.

An \emph{$\Fqm$-linear sum-rank-metric code} $\mycode{C}$ is an $\Fqm$-subspace of $\Fqm^{n}$.
It has length $n$ (with respect to a length partition $\n$), dimension $k \defeq \dim_{q^m}(\mycode{C})$ and minimum
(sum-rank) distance $d \defeq \min \{\SumRankDist(\x,\y) : \x, \y \in \mycode{C}\}$.
To emphasize its parameters, we write $\mycode{C}[\n, k, d]$ in the following.

\subsection{Interleaved Sum-Rank-Metric Codes and Channel Model}
A (vertically) $s$-interleaved code is a direct sum of $s$ codes of the same length $n$.
In this paper we consider \emph{homogeneous} interleaved codes, i.e. codes obtained by interleaving codewords of \emph{a single} constituent code.

\begin{definition}[Interleaved Sum-Rank-Metric Code]
    Let $\mycode{C}[\n,k,d]\subseteq\Fqm^n$ be an $\Fqm$-linear sum-rank-metric code of length $n$ with length partition $\n=(n_1,n_2,\dots,n_\shots)\in\NN^\shots$ and minimum sum-rank distance $d$.
    Then the corresponding (homogeneous) $\intOrder$-interleaved code is defined as
    \begin{equation*}
        \mycode{IC}[\intOrder;\n,k,\dminsr]
        \defeq
        \left\{
        \begin{pmatrix}
            \c_1
            \\[-4pt]
            \vdots
            \\[-4pt]
            \c_\intOrder
        \end{pmatrix}:
        \c_j=(\c_j^{(1)} \, | \, \dots \, | \, \c_j^{(\shots)}) \in \mycode{C}[\n,k,\dminsr]
        \right\}
        \subseteq\Fqm^{\intOrder \times n}.
    \end{equation*}
\end{definition}

Each codeword $\C\in\mycode{IC}[\intOrder;\n,k,d]$ can be written as
\begin{equation*}
    \mat{C}=
    \left(
    \begin{array}{c|c|c|c}
        \c_1^{(1)} & \c_1^{(2)} & \dots & \c_1^{(\shots)}
        \\
        \vdots & \vdots & \ddots & \vdots
        \\
        \c_\intOrder^{(1)} & \c_\intOrder^{(2)} & \dots & \c_\intOrder^{(\shots)}
    \end{array}
    \right)\in\Fqm^{\intOrder \times n}
\end{equation*}
or equivalently as
\begin{equation*}
    \mat{C}=
    (\vec{C}^{(1)} \mid \mat{C}^{(2)} \mid \dots \mid \mat{C}^{(\shots)})
\end{equation*}
where
\begin{equation}
    \label{eq:comp_mat_ILRS}
    \vec{C}^{(i)}\defeq
    \begin{pmatrix}
        \c_1^{(i)}
        \\
        \c_2^{(i)}
        \\
        \vdots
        \\
        \c_\intOrder^{(i)}
    \end{pmatrix}
    \in\Fqm^{\intOrder\times n_i}
\end{equation}
for all $i\in[1:\shots]$.

As a channel model we consider the additive sum-rank channel
\begin{equation}
    \Y = \C + \E
\end{equation}
where
\begin{equation}
    \E=(\E^{(1)} \, | \, \E^{(2)} \, | \, \dots \, | \, \E^{(\shots)} )\in\Fqm^{\intOrder \times n}
\end{equation}
with $\E^{(i)}\in\Fqm^{\intOrder \times n_i}$ and $\rk_q(\E^{(i)})=t_i$ for all $i\in[1:\shots]$ is an error matrix with $\SumRankWeight(\E)=\sum_{i=1}^{\shots}t_i=t$.

\section{Decoding of High-Order Interleaved Sum-Rank-Metric Codes} \label{sec:decoder}

%\begin{itemize}
%    \item Statements and algorithm
%
%    \item Estimate failure probability if no condition on the error is enforced \todo{FH}
%        \begin{itemize}
%            \item There should be an intuition that particular patterns like e.g. well-balanced error rank partitions are more likely to have $\Fqm$-full-rank.
%
%            \item Make some simulations to evaluate the tightness of the bound
%        \end{itemize}
%\end{itemize}

In this section, we propose a \MeKap-like decoder for the sum-rank metric,
that is a generalization of the decoders proposed in~\cite{metzner1990general, puchinger2019decoding, renner2021decoding}.
Similar to the Hamming- and the rank-metric case, the proposed decoder works for errors of sum-rank weight $t$ up to $d-2$ that satisfy the following conditions:
\begin{itemize}
    \item \emph{High-order condition:} The interleaving order $\intOrder \geq t$,
    \item \emph{Full-rank condition:} Full $\Fqm$-rank error matrices, i.e., $\rkqm(\E) = t$.
\end{itemize}
%\fh{Do we maybe want to write full sentences in the conditions here?}
%\tj{I think it is fine as is}
%\fh{Alright. :-)}
Note that the full-rank condition implies the high-order condition since the $\Fqm$-rank of a matrix $\E \in \Fqm^{\intOrder \times n}$ is at most $\intOrder$.
%The full-rank condition is satisfied for most error matrices $\E$ of sum-rank weight $t$ (see also~\cite{metzner1990general, renner2021decoding}).
%\fh{Do we want to give any reasoning for that? E.g. "by simulations" or "as one can see by simple counting arguments" etc.}
%\tj{To be honest, I think we should remove this sentence, since we don't know at this point. We know with the upcoming results about the error probability and because I computed/simulated it for some parameters. But this is future work.
%Deducting that it is the case for the sum-rank metric because it is the case fo the Hamming and rank metric is a bit vague in my opinion here. What do you think?}
%\hb{If you know it from simulations, I would put it here, since this makes the result stronger.}
%\fh{I agree that \emph{if} we want to keep this, we should argue more (e.g. with your simulations) and already touch on future work. But I'm also totally fine with leaving it out -- it's definitely the simpler option and noone will actually miss it, I guess.}
%\tj{Yeah, but the "simulations" are just prelimanry results. I'm also not sure if that is true for ALL parameters. So either way it would be vague assumptions at this point. I'm more of the opinion of better saying nothing than saying something vauge or even wrong. But I leave it to you :P}

Throughout this section we consider a homogeneous $\intOrder$-interleaved sum-rank-metric code $\mycode{IC}[\intOrder;\n,k,d]$ over a field $\Fqm$ with a constituent code $\mycode{C}[\n,k,d]$ defined by a parity-check matrix 
\begin{equation}
    \H = (\H^{(1)}\,|\,\H^{(2)}\,|\,\ldots\,|\,\H^{(\shots)})\in\Fqm^{(n-k)\times n}
\end{equation}
with $\H^{(i)}\in\Fqm^{(n-k) \times n_i}$.
The goal is to recover a codeword $\C\in\mycode{IC}[\intOrder;\n,k,d]$ from the matrix
\begin{equation*}
    \Y = \C + \E \in \Fqm^{\intOrder \times n}
\end{equation*}
that is corrupted by an error matrix $\E$ of sum-rank weight $\SumRankWeight(\E)=t$ assuming \emph{high-order} and \emph{full-rank} condition.
%\fh{Do we want to incorporate the conditions here once more? E.g. "...of sum-rank weight $\SumRankWeight(\E)=t \leq s$ and of $\Fqm$-rank $t$."}
%\tj{I added high order and full rank condition in words. Remove it or leave it as you think fits best.}
%\fh{I think it's good, thanks for adding it.}

As the Metzner--Kapturowski algorithm and its adaptation to the rank metric, the presented decoding algorithm consists of two steps.
%\fh{colon since we name them afterwards?}
%\tj{I don't like colons inline in a sentence. Looks weird. If using colons then like above with an itemize following. But an itemize here would be overkill. So I would just leave it as it is.}
%\fh{Alright, fine by me. :-)}
The decoder first determines the error support from the syndrome matrix $\S = \H\Y^\top$.
Secondly, erasure decoding is performed to recover the error $\E$ itself.

\subsection{The Error Support}

The error matrix $\E$ can be decomposed as
\begin{equation}\label{eq:error_decomp}
    \E = \A \B
\end{equation}
where $\A = (\A^{(1)} \, | \, \A^{(2)} \, | \, \dots \, | \, \A^{(\shots)} )\in\Fqm^{\intOrder \times t}$ with $\A^{(i)}\in\Fqm^{\intOrder\times t_i}$ and $\rkq(\A^{(i)})=t_i$
%\fh{We have never said what $t_i$ is for $\E$ (except for the implicit mention in the definition of the sum-rank weight in the preliminaries). I think we should add that in the previous section.}
%\tj{In the last sentence of the preliminaries it literally says $t_i=\rkq(\E^{(i)})$. If that is implicit, what is explicit then? :)}
%\fh{Oh sorry, I checked the definition of the sum-rank metric and not the channel model. :D}
and
\begin{equation}\label{eq:def_B}
    \B = \diag{(\B^{(1)},\dots,\B^{(\shots)})} \in \Fq^{t \times n}
\end{equation}
with $\B^{(i)}\in\Fq^{t_i \times n_i}$ and $\rkq(\B^{(i)})=t_i$ for all $i\in[1:\shots]$ (see~\cite[Lemma 10]{puchinger2020generic}).
The rank support of one block $\E^{(i)}$ is defined as
\begin{equation*}
    \RankSupp\left(\E^{(i)}\right)\defeq\Rowspace{\E^{(i)}}=\Rowspace{\B^{(i)}}.
\end{equation*}
%\fh{In my opinion, the above equation is quite random here. I understand that we use it later when we compare the different settings. However, I'd suggest to either delete it here and only define it when we need it or to keep it here but with some more context.}
The sum-rank support for the error $\E$ with sum-rank weight $t$ is then defined as
\begin{align}\label{eq:sumranksupport}
    \SumRankSupp(\E)
    \defeq
    &\RankSupp{\left(\E^{(1)}\right)} \times \RankSupp{\left(\E^{(2)}\right)} \times \dots \times \RankSupp{\left(\E^{(\shots)}\right)}
    \\
    =&\Rowspace{\B^{(1)}} \times \Rowspace{\B^{(2)}} \times \dots \times \Rowspace{\B^{(\shots)}}.
\end{align}

%\tj{Maybe introduce a notion for basis of $\SumRankSupp(\E)$}
%\tj{Follow up: I'm not sure anymore if we need this.}
%\fh{I think both options are fine. We don't really \emph{need} it but I think one can argue that the notion still makes sense since we (implicitly) use it throughout the paper. Feel free to decide. :-)}

The following result from~\cite{puchinger2020generic} shows how the error matrix $\E$ can be reconstructed from the sum-rank support and the syndrome matrix $\S$.

\begin{lemma}[Column-Erasure Decoder~{\cite[Theorem 13]{puchinger2020generic}}]\label{lem:col_erasure_decoding}
    Let $t<d$ and $\B = \diag{(\B^{(1)},\dots,\B^{(\shots)})}\in\Fq^{t\times n}$ be a basis of the row space of the error matrix $\E\in\Fqm^{\intOrder\times n}$
%    such that $\B^{(i)}$ is a basis of the row space of $\E^{(i)}$ for all $i\in[1:\ell]$
    and $\S = \H\E^\top \in \Fqm^{(n-k)\times \shots}$ be the corresponding syndrome matrix.
    Then, the error is given by $\E = \A\B$ with $\A$ being the unique solution of the linear system
    \begin{equation}
        \S = (\H\B^{\top})\A^\top
    \end{equation}
    and $\E$ can be computed in $\oh{(n-k)^3 m^2}$ operations over $\Fq$.
\end{lemma}

%\begin{proof}
%    See Proof of Theorem~13 in~\cite{puchinger2020generic}.
%\end{proof}
%
%\tj{Maybe move this to preliminaries?}

\subsection{Recovering the Error Support}

In the following we show how to recover the sum-rank support $\SumRankSupp(\E)$ of the error $\E$
given the syndrome matrix
\begin{equation}\label{eq:syndromeMatrix}
        \S = \H\Y^\top
        =\H\E^\top
        =\sum_{i=1}^{\shots}\H^{(i)}(\E^{(i)})^\top
\end{equation}
and the parity-check matrix $\H$ of the sum-rank-metric code $\mycode{IC}[\intOrder;\n,k,d]$.
Let $\P\in\Fqm^{(n-k)\times(n-k)}$ with $\rkqm(\P)=n-k$ be such that $\P\S = \REF(\S)$.
Further, let $\Hsub$ be the rows of $\P\H$ corresponding to the zero rows in $\P\S$, i.e. we have
    \begin{equation*}
        \P\S =
        \begin{pmatrix}
            \S'
            \\
            \0
        \end{pmatrix}
        \quad\text{and}\quad
        \P\H=
        \begin{pmatrix}
            \H'
            \\
            \Hsub
        \end{pmatrix}
    \end{equation*}
where $\S'$ and $\H'$ have the same number of rows.
Since $\P$ performs $\Fqm$-linear row operations on $\H$, the $\shots$ blocks of $\P\H$ are preserved, i.e. we have that
\begin{equation}
    \Hsub = \left( \Hsub^{(1)} \, | \, \Hsub^{(2)} \, | \, \dots \, | \, \Hsub^{(\shots)} \right).
\end{equation}

The following lemma is a generalization of~\cite[Lemma~3]{puchinger2019decoding} to the sum-rank metric.
%\fh{colon because of "the following"?}\tj{no}
\begin{lemma}
    \label{lem:row_space_H_sub}
    Let $\H=(\H^{(1)} \, | \, \H^{(2)} \, | \, \dots \, | \, \H^{(\shots)})\in\Fqm^{(n-k) \times n}$ be a parity-check matrix of a sum-rank-metric code $\mycode{C}$ and let $\S=\H\E^\top\in\Fqm^{(n-k) \times \intOrder}$ be the syndrome matrix of an error
    \begin{equation*}
        \E=(\E^{(1)} \, | \, \E^{(2)} \, | \, \dots \, | \, \E^{(\shots)})\in\Fqm^{\intOrder \times n}
    \end{equation*}
    of sum-rank weight $\SumRankWeight(\E)=t < n-k$ where $\E^{(i)}\in\Fqm^{\intOrder \times n_i}$ with $\rk_q(\E^{(i)})=t_i$ for all $i\in[1:\shots]$.
    Let $\P\in\Fqm^{(n-k) \times (n-k)}$ be a matrix with $\rk_{q^m}(\P)=n-k$ such that $\P\S$ is in row-echelon form.
    Then, $\P\S$ has at least $n-k-t$ zero rows.
%    \fh{Maybe use ``$\P\S$ has at least $n-k-t$ zero rows'' since it may be confusing to read that rows are equal to a number.}
    Let $\Hsub$ be the submatrix of $\P\H$ corresponding to the zero rows in $\P\S$.
    Then we have that
    \begin{equation*}
        \label{eq:row_space_H_sub}
        \RowspaceFqm{\Hsub}=\rker(\E)_{q^m} \cap \mycode{C}^\perp
        \;\Longleftrightarrow\;
        \RowspaceFqm{\Hsub}=\rker(\E)_{q^m} \cap \RowspaceFqm{\H}.
    \end{equation*}
\end{lemma}

\begin{proof}
    Since $\E^{(i)}$ has $\Fq$-rank $t_i$, its $\Fqm$-rank is at most $t_i$ for all $i\in[1:\ell]$. Since $t=\sum_{i=1}^{\ell} t_i$, $\E$ has at most $\Fqm$-rank $t$ as well.
    Hence, the $\Fqm$-rank of $\S$ is at most $t$ and thus at least $n-k-t$ of the $n-k$ rows of $\P\S$ are zero.

    The rows of $\P\H$ and therefore also the rows of $\Hsub$ are in the row space of $\H$, i.e. in the dual code $\mycode{C}^{\perp}$.
    Since $\Hsub\E^{\top}=\0$ the rows of $\Hsub$ are in the kernel of $\E$. It is left to show that the rows of $\Hsub$ span the entire intersection space.
    Write
        \begin{equation*}
        \P\S =
        \begin{pmatrix}
            \S'
            \\
            \0
        \end{pmatrix}
        \quad\text{and}\quad
        \P\H=
        \begin{pmatrix}
            \H'
            \\
            \Hsub
        \end{pmatrix}
    \end{equation*}
    where $\S'$ and $\H'$ have the same number of rows and $\S'$ has full $\Fqm$-rank.
    Let $\v = (\v_1, \v_2) \in \Fqm^{n-k}$ and let
    \begin{equation*}
        \h = \v    \cdot  \begin{pmatrix}
            \H'
            \\
            \Hsub
        \end{pmatrix}
    \end{equation*}
    be a vector in the row space of $\P\H$ and in the kernel of $\E$.
    Since $\Hsub\E^{\top}=\bm{0}$ we have that $\bm{0} = \h \E^{\top} = \v_1 \H' \E^{\top} = \v_1 \S'$.
    This implies that $\v_1 = \0$ since the rows of $\S'$ are linearly independent and thus $\h$ is in the row space of $\Hsub$.
\qed
\end{proof}

Lemma~\ref{lem:Fqm_kernel_E_B} shows that the kernel $\rker{(\E)}_{q^m}$ of the error $\E$ is connected with the kernel of the matrix $\B$ if the $\Fqm$-rank of the error is $t$, i.e. if the full-rank condition is satisfied.

\begin{lemma}
    \label{lem:Fqm_kernel_E_B}
    Let $\E=(\E^{(1)} \, | \, \E^{(2)} \, | \, \dots \, | \, \E^{(\shots)})\in\Fqm^{\intOrder \times n}$ be an error of sum-rank weight $\SumRankWeight(\E)=t$ where $\E^{(i)}\in\Fqm^{\intOrder \times n_i}$ with $\rk_q(\E)=t_i$ for all $i\in[1:\shots]$.
%    If $\intOrder \geq t$ (high-order condition) and $\rk_{q^m}(\E)=t$ (full-rank condition), then
    If $\rk_{q^m}(\E)=t$ (full-rank condition), then
    \begin{equation}
        \label{eq:Fqm_kernel_E_B}
        \rker(\E)_{q^m}=\rker(\B)_{q^m}
    \end{equation}
    where $\B\in\Fq^{t \times n}$ is any basis for the $\Fq$-row space of $\E$ of the form~\eqref{eq:def_B}.
%    \fh{This is actually the thing we talked about today. Maybe we should try to use similar formulations here and in Lemma 1 in the end.}
%    \tj{I think we should not overcomplicate things here. We state that $\B$ must be of the form~\eqref{eq:def_B} which is of block-diagonal structure.}
%    \fh{That's totally fine. I just wanted to leave a comment in case we introduce a different notation to keep in mind to change it here, too. But since we didn't, nothing to do here. :-)}
    Further, it holds that
    \begin{equation}
        \rker(\E^{(i)})_{q^m}=\rker(\B^{(i)})_{q^m},\quad\forall i\in[1:\shots].
    \end{equation}
\end{lemma}

\begin{proof}
    Let $\E$ have $\Fqm$-rank $t$ and let $\E=\A\B$ be a decomposition of the error as in~\eqref{eq:error_decomp} such that $\E^{(i)}=\A^{(i)}\B^{(i)}$ for all $i\in[1:\shots]$.
    Since $\rk_{q^m}(\E)=t$ implies that $\rk_{q^m}(\A)=t$, we have that $\rker(\A)_{q^m}=\{\0\}$.
    Hence, for all $\v\in\Fqm^{n}$, $(\A\B)\v^\top=\0$ if and only if $\B\v^\top=\0$ which is equivalent to
    \begin{equation}
        \label{eq:ker_B}
        \rker(\E)_{q^m}=\rker(\A\B)_{q^m}=\rker(\B)_{q^m}.
    \end{equation}
    Assume a vector $\v=(\v^{(1)} \,|\, \v^{(2)} \,|\, \dots \,|\, \v^{(\shots)})\in\Fqm^n$ and let $\v^{(i)}\in\Fqm^{n_i}$ be any element in $\rker(\B^{(i)})_{q^m}$.
%    \fh{overfull hbox in previous line}\tj{I tried rewrite the sentence. Is overfull hbox fixed now?}
%    \fh{Yes, that's great. But is the meaning still the same? We would need to choose the $\v^{(i)}$ in the kernel of $\B^{(i)}$ and not $\B$, right?}
    Due to the block-diagonal structure of $\B$ (see~\eqref{eq:def_B}) we have that
    \begin{equation}
        \B\v^\top = \0
        \quad\Longleftrightarrow\quad
        \B^{(i)}(\v^{(i)})^\top = \0, \quad\forall i\in[1:\shots]
    \end{equation}
    which is equivalent to
    \begin{equation}
        \label{eq:ker_Bi}
        \v\in\rker(\B)_{q^m}
        \quad\Longleftrightarrow\quad
        \v^{(i)}\in\rker(\B^{(i)})_{q^m},\quad\forall i\in[1:\shots].
    \end{equation}
    Combining~\eqref{eq:ker_B} and~\eqref{eq:ker_Bi} yields the result.
\qed
\end{proof}

Combining Lemma~\ref{lem:row_space_H_sub} and Lemma~\ref{lem:Fqm_kernel_E_B} finally allows us to recover the sum-rank support of $\E$.

\begin{theorem}
    \label{thm:sum_rank_support_recovery}
    Let $\E=(\E^{(1)} \, | \, \E^{(2)} \, | \, \dots \, | \, \E^{(\shots)})\in\Fqm^{\intOrder \times n}$ be an error of sum-rank weight $\SumRankWeight(\E)=t \leq d-2$ where $\E^{(i)}\in\Fqm^{\intOrder \times n_i}$ with $\rk_q(\E^{(i)})=t_i$ for all $i\in[1:\shots]$.
    If $\intOrder \geq t$ (high-order condition) and $\rk_{q^m}(\E)=t$ (full-rank condition), then
    \begin{equation}
        \label{eq:row_space_Ei}
        \Rowspace{\E^{(i)}}=\rker\left(\ext(\Hsub^{(i)})\right)_q, \quad\forall i\in[1:\shots].
    \end{equation}
\end{theorem}

\begin{proof}

%    From Lemma~\ref{lem:row_space_H_sub} we get
%        \begin{equation}\label{eq:proofeq1}
%            \RowspaceFqm{\Hsub}=\rker(\E)_{q^m} \cap \mycode{C}^\perp.
%        \end{equation}
%        Since $\rk_{q^m}(\E)=t$ and $\intOrder \geq t$ we get from Lemma~\ref{lem:Fqm_kernel_E_B} (see~\eqref{eq:Fqm_kernel_E_B}) that
%        \begin{equation}\label{eq:proofeq2}
%            \rker(\E)_{q^m}=\rker(\B)_{q^m}.
%        \end{equation}
%        By combining~\eqref{eq:proofeq1} and~\eqref{eq:proofeq2} we get
%        \begin{equation}\label{eq:combine_lemmas}
%            \RowspaceFqm{\Hsub}=\rker(\B)_{q^m} \cap \RowspaceFqm{\H}.
%        \end{equation}
%    \fh{Maybe add something like "which we will use later on." to make clear why we write this down here.}
%    \fh{Why do we have the above part of the proof? As far as I understand, it doesn't prove any statement from the theorem. Do we use it anywhere?
%    --- Okay, I get it now after rereading the rank-metric paper. Can we somehow make clearer that these results imply that the rest of our proof actually proves the statements from the theorem?}
    In the following, we prove that the $\Fq$-row space of the extended $\Hsub$ instead of the $\Fqm$-row space of $\Hsub$ is equal to the $\Fq$-kernel of $\B$, i.e.,
    \begin{equation}
        \Rowspace{\ext(\Hsub)} = \rker(\B)_q.
    \end{equation}

%    \fh{Good, but I think we're still missing the fact that we show it blockwise and then argue via the dual spaces. What do you think? (We could also add it at the end of the proof if it gets too long here.)}
%    \tj{Maybe you can add the blockwise part if you want to make it more clear.}
%    \fh{Here's a suggestion (that doesn't necessarily need your comment but I leave it up to you if we still want to include it):}
    Recall that $\Rowspace{\E^{(i)}} = \Rowspace{\B^{(i)}}$ holds for all $i \in [1:\shots]$ according to the definition of the error decomposition~\eqref{eq:error_decomp}.
    With this in mind, the statement of the theorem is equivalent to showing $\rker(\B^{(i)})_q = \Rowspace{\ext(\Hsub^{(i)})}$ for all $i\in[1:\shots]$ since $\Rowspace{\B^{(i)}}^{\perp} = \rker(\B^{(i)})_q$ and $\rker\left(\ext(\Hsub^{(i)})\right)_q^{\perp} = \Rowspace{\ext(\Hsub^{(i)})}$ hold.

    First we show that $\Rowspace{\ext(\Hsub)} \subseteq \rker(\B)_q$ which, due to the block-diagonal structure of $\B$, implies that $\Rowspace{\ext(\Hsub^{(i)})} \subseteq \rker(\B^{(i)})_q$ for all $i\in[1:\shots]$.
    Let $\v=(\v^{(1)} \,|\, \v^{(2)} \,|\, \dots \,|\, \v^{(\shots)})\in\Fqm^n$ with $\v^{(i)}\in\Fqm^{n_i}$ for all $i\in[1:\shots]$ be any element in the $\Fqm$-linear row space of $\Hsub$.
    Then, by~\cite[Lemma~5]{puchinger2019decoding} we have that each row $\v_j$ for $j\in[1:m]$ of $\ext(\v)$ is in $\Rowspace{\ext(\Hsub)}$ which implies that $\v_j^{(i)}\in\Rowspace{\ext(\Hsub^{(i)})}$ for all $i\in[1:\shots]$.
    By Lemma~\ref{lem:Fqm_kernel_E_B} we have that $\v\in\rker(\B)_{q^m}$, i.e. we have
    \begin{equation}
        \B\v^\top=\0
        \quad\Longleftrightarrow\quad
        \B^{(i)}(\v^{(i)})^\top=\0,\quad \forall i\in[1:\shots]
    \end{equation}
    where the right-hand side follows from the block-diagonal structure of $\B$.
    Since the entries of $\B$ are from $\Fq$, we have that
    \begin{equation}
        \label{eq:kernel_B}
        \ext(\B\v^\top)=\B\ext(\v)^\top=\0
    \end{equation}
    which implies that $\v\in\rker(\B)_q$ and thus $\Rowspace{\ext(\Hsub)} \subseteq \rker(\B)_q$.
    Due to the block-diagonal structure of $\B$ we get from~\eqref{eq:kernel_B} that
    \begin{equation}
        \label{eq:kernel_B_i}
        \ext(\B^{(i)}(\v^{(i)})^\top)=\B^{(i)}\ext(\v^{(i)})^\top=\0,
        \quad\forall i\in[1:\shots]
    \end{equation}
    which implies that $\v_j^{(i)}\in\rker(\B^{(i)})_q$ for all $i\in[1:\shots]$ and $j\in[1:m]$.
    Therefore, we have that $\Rowspace{\ext(\Hsub^{(i)})} \subseteq \rker(\B^{(i)})_q$, for all $i\in[1:\shots]$.

     Next, we show that $\rker(\B^{(i)})_q = \Rowspace{\ext(\Hsub^{(i)})}$ for all $i\in[1:\shots]$ by showing that
        \begin{align}
            r_i&\defeq\dim\left(\Rowspace{\ext(\Hsub^{(i)})}\right) = n_i - t_i,\quad\forall i\in[1:\shots].
        \end{align}
        Since $\Rowspace{\ext(\Hsub^{(i)})} \subseteq \rker(\B^{(i)})_q$ we have that $r_i > n_i - t_i$ is not possible for all $i\in[1:\shots]$.
%    \fh{Should be "any" I think.}\tj{I think its all here. The any part comes later with the contradiction. But I'll think about it.} \fh{I think the confusion comes from the negation in the sentence. As far as I understand "not possible for all" means it might be possible for some. But that's really a language detail for which we could use a native speaker's opinion...}

       In the following we show that $r < n-t$ is not possible and therefore $r_i = n_i - t_i$ holds for all $i=[1:\shots]$.
        Let $\{\h_1, \ldots, \h_r\}\subseteq\Fq^{n}$ be a basis for $\Rowspace{\ext(\Hsub)}$ and define
        \begin{equation}
            \H_{b} =
            \begin{pmatrix}
            {\h_1}
                \\
                {\h_2}
                \\
                \vdots
                \\
                {\h_{r}}
            \end{pmatrix}
            \in\Fq^{r \times n}
        \end{equation}
        with $\h_j = (\h_j^{(1)}\, | \, \h_j^{(2)}\, | \, \dots\, | \, \h_j^{(\ell)})\in\Fq^n$
%    \fh{Is there a reason why we don't use $\mid$ to divide the blocks here?}\tj{I added them but then to be consistent, one should add also bars when stacking matrices vertically? or both like when introducing $\J$}
%    \fh{I thought about the stacking, too. But I think it's fine if we don't do it for that. And for $\J$: I accepted for myself that we only do it when we have the division into $\shots$ blocks but yeah, it's not perfect. I don't know what's the best approach.}
    where $\h_j^{(i)} \in \Rowspace{\ext(\Hsub^{(i)})}$ for $j\in[1:r]$ and $i\in[1:\shots]$.
        Also define
        \begin{equation}
        {\H_{b}^{(i)}}
            =
            \begin{pmatrix}
            {\h_1^{(i)}}
                \\
                {\h_2^{(i)}}
                \\
                \vdots
                \\
                {\h_{r}^{(i)}}
            \end{pmatrix}
            \in\Fq^{r \times n_i},
            \quad \forall i\in[1:\shots].
        \end{equation}
        % \tj{Note: it is $r \times n_i$ now not $r_i \times n_i$!}
%        \fh{Can we integrate this equation into the last one?}

        By the basis-extension theorem, there exist matrices ${\B^{(i)}}''\in\Fq^{(n_i-t_i) \times n_i}$ such that the matrices
        \begin{equation}
        {\B^{(i)}}
            '\defeq
            \begin{pmatrix}
                \left(\B^{(i)}\right)^\top \, | \,  \left({\B^{(i)}}''\right)^\top
            \end{pmatrix}\in\Fq^{n_i \times n_i}
        \end{equation}
        have $\Fq$-rank $n_i$ for all $i\in[1:\shots]$.

%        \fh{General remark for the rest of the proof: Is the notation $\H'$ confusing because we also use it for the division of $\H$ into $\H'$ and $\H_{sub}$?}\tj{Yeah good point. I'm thinking about changing $H'$ for the division of $\H$ into sth. else. Maybe $\H_{\text{sup}}$? :D Probably also changing $\S'$ accordingly.}
%        \fh{$\H_{sup}$ is the best idea ever. :D What about $\bar{\H}$? And just two thoughts about whether we should change the notation here in this proof or the overall notation for the division of $\H$: currently, the division sticks to the same notation as the rank-metric paper. I think that's good. Moreover, changing the proof only is less work and less prone to introducing inconsistencies.}\tj{They make the same mistake in the rank-metric paper. To be consistent is good but not if it means to copy mistakes. Also it doesn't hurt to give the paper our own flavor and not copy pasta too much. I'd rather change the above part so either $\H_\text{sup}$ or with bar above. But just my opinion.}
%        \fh{We would make it consistent either way. But you can of course also change the division of $\H$, why not. :-) I just wouldn't use $\H_{sup}$ and $\H_{sub}$ because it's too close and thus probably very confusing to the readers.}
        Next define ${\check{\H}^{(i)}} = \H_b^{(i)} {\B^{(i)}}' \in \Fq^{r \times n_i}$ for all $i\in[1:\shots]$ and
        \begin{equation}
            \check{\H} \defeq \begin{pmatrix} {\check{\H}^{(1)}}
                       \, | \,  {\check{\H}^{(2)}} \, | \, \dots \, | \,  {\check{\H}^{(\ell)}}
            \end{pmatrix} = \H_b \cdot \diag{\left({\B^{(1)}}',\ldots,{\B^{(\ell)}}'\right)}.
        \end{equation}
%        \fh{We usually use $\mid$ to separate matrices. Should we add that for $\H'$?}

        Since $\h_1^{(i)}, \h_2^{(i)},\dots,\h_{r}^{(i)}$ are in the right $\Fq$-kernel of $\B^{(i)}$ (see~\eqref{eq:kernel_B_i}) we have that
        \begin{align}
        {\check{\H}^{(i)}} =
            \begin{pmatrix}
                0 & \dots & 0 & {\check{h}_{1,t_i+1}^{(i)}} & \dots & {\check{h}_{1,n_i}^{(i)}}
                \\
                \vdots & \ddots & \vdots & \vdots & \ddots & \vdots
                \\
                0 & \dots & 0 & {\check{h}_{r,t_i+1}^{(i)}} & \dots & {\check{h}_{r,n_i}^{(i)}}
            \end{pmatrix}
        \end{align}
        for all $i\in[1:\ell]$ and thus $\check{\H}$ has at least $t = \sum_{i=1}^{\ell} t_i$ all-zero columns.

        By the assumption that $r < n- t$ it follows that $r_i < n_i - t_i$ holds for at least one block.
        Without loss of generality assume that this holds for the $\ell$-th block, i.e. we have $r_\ell < n_\ell - t_\ell$.
        Then there exists a full-rank matrix
        \begin{equation}
%            \J = \begin{pmatrix}
%                     \I_{t_\ell} & \bm{0}         \\
%                     \bm{0}      & \widetilde{\J}
%            \end{pmatrix}
            \J = \mleft(\begin{array}{c|c}
                 \I_{t_\shots} & \bm{0} \\\hline\\[-2.0ex]
                 \bm{0} & \widetilde{\J}
        \end{array}\mright)
            \in\Fq^{n_\shots \times n_\shots}
        \end{equation}
        with $\widetilde{\J} \in \Fq^{(n_\ell-t_\ell)\times(n_\ell-t_\ell)}$ such that the matrix
        \begin{equation}
            \label{eq:hastplusonezero}
            \widetilde{\H} = \check{\H} \cdot \diag{\left(\I_{n_1}, \ldots,\I_{n_{\shots - 1}},\J\right)}
        \end{equation}
        % with
        % \begin{equation}
        %     \J = \begin{pmatrix}
        %              \I_{t_\ell} & \bm{0}         \\
        %              \bm{0}      & \widetilde{\J}
        %     \end{pmatrix}
        % \end{equation}
        % and $\widetilde{\H}$
        has at least $t+1$ all-zero columns.

        Define $\D \defeq \diag{\left({\B^{(1)}}',\ldots, {\B^{(\ell-1)}}',{\B^{(\ell)}}' \J\right)} \in \Fq^{n \times n}$ which has full $\Fq$-rank $n$.
        Then we have that $\widetilde{\H} = \H_b \cdot \D$.
        Since $\D$ has full $\Fq$-rank $n$, the submatrix $\D' \defeq \D_{[1:n],\mathcal{I}} \in \Fq^{n \times (t+1)}$ has $\Fq$-rank $t+1$, where
%        \begin{equation}
%            \mathcal{I} = \{1,\ldots,t_1, n_1+1,\ldots,n_1+t_2,n_{\ell-2}+1,\ldots,n_{\ell-2}+t_{\ell-1},n_{\ell-1}+1,\ldots,n_{\ell-1}+t_{\ell}+1\}.
%        \end{equation}
%        \fh{This should be the same as
        \begin{equation}
            \mathcal{I} = [1:t_1] \cup [n_1+1:n_1+t_2] \cup [n_{\ell-2}+1:n_{\ell-2}+t_{\ell-1}] \cup [n_{\ell-1}+1:n_{\ell-1}+t_{\ell}+1]
        \end{equation}
%        which is shorter. :)
%        }
        By~\eqref{eq:hastplusonezero} it follows that
        \begin{equation}\label{eq:mustbeallzero}
            \h_j \cdot \D' = \0 \in \Fq^{t+1} %(0,\ldots,0) \in \Fq^{t+1}
        \end{equation}
%        \fh{We've used $\0$ to indicate an all-zero vector before. Do we want to do it here as well?}
        for all $j\in[1:r]$.
        Since $\H \in \Fqm^{(n-k) \times n}$ is a parity-check matrix of an $[\n,k,d]$ code it has at most $d-1$ $\Fqm$-linearly dependent columns (see~\cite[Lemma~12]{puchinger2020generic}).
        Since by assumption $t+1 \leq d-1$ and $\rkq(\D') = t+1$ we have that $\rkqm (\H \D') = t + 1$.
        Thus, there exists a vector $\g \in \RowspaceFqm{\H}$ such that
        \begin{equation}
            \g \D' =  (\0, g_{t+1}') \in \Fqm^{t+1}
        \end{equation}
%        \fh{Same as above: We've used $\0$ to indicate an all-zero vector before. Do we want to do it here as well?}
        with $g_{t+1}'\neq 0$. Since the first $t$ positions of $\g\D'$ are equal to zero we have that $\g \in \RowspaceFqm{\Hsub}$.
    Expanding the vector $\g\D'$ over $\Fq$ gives
        \begin{equation}
            \ext(\g)\D' = \begin{pmatrix}
                              \0 & g_{1,t+1}' \\
                              \0 & g_{2,t+1}' \\
                              \vdots & \vdots \\
                              \0 & g_{m,t+1}'
            \end{pmatrix} \in \Fq^{m\times(t+1)}
        \end{equation}
        where $\ext(g_{t+1}') = (g_{1,t+1}',g_{2,t+1}',\ldots,g_{m,t+1}')^\top \in \Fq^{m\times 1}$.
        Since $g_{t+1}'\neq 0$ there exists at least one row with index $\iota$ in $\ext(\g_{t+1}')$ such that $g_{\iota, t+1}'\neq 0$.
        Let $\g_\iota$ be the row in $\ext(\g)$ for which $\g_\iota\D'$ is not all-zero.
        This leads to a contradiction according to~\eqref{eq:mustbeallzero}.
        Thus $r < n - t$ is not possible and leaves $r = n - t$ and therefore also $r_i = n_i - t_i$ for all $i\in[1:\shots]$ as the only valid option.
        \qed
\end{proof}

\subsection{A \MeKap-like Decoding Algorithm}

Using Theorem~\ref{thm:sum_rank_support_recovery} we can formulate an efficient decoding algorithm for high-order interleaved sum-rank-metric codes.
The algorithm is given in Algorithm~\ref{alg:decode_high_order_int_sum_rank} and proceeds similar to the \MeKap(-like) decoding algorithms for Hamming- or rank-metric codes.
As soon as $\Hsub$ is computed from the syndrome matrix $\S$, the rank support  of each block can be recovered independently using the results from Theorem~\ref{thm:sum_rank_support_recovery}.
%\fh{Suggestion: "the rank support of each block can be recovered independently using ...".}, which \fh{I'd start a new sentence here.}
This corresponds to finding a matrix $\B^{(i)}$ with $\rk_q(\B^{(i)})=t_i$ such that $\ext(\Hsub^{(i)})(\B^{(i)})^\top = \0$ for all $i\in[1:\ell]$ (see~\eqref{eq:row_space_Ei}).

\begin{algorithm}[ht!]
    \setstretch{1.35}
    \caption{Decoding High-Order Interleaved Sum-Rank-Metric Codes}\label{alg:decode_high_order_int_sum_rank}
    \SetKwInOut{Input}{Input}\SetKwInOut{Output}{Output}

    \Input{Parity-check matrix $\H$, Received word $\Y$}

    \Output{Transmitted codeword $\C$}
    \BlankLine

    $\S \gets \H\Y^\top \in \Fqm^{(n-k) \times \intOrder}$\label{step:syndromeMatrix}
    \\
    Compute $\P\in\Fqm^{(n-k) \times (n-k)}$ s.t. $\P\S=\REF(\S)$\label{step:transfP}
    \\
    $\Hsub=\left(\Hsub^{(1)} \, | \, \Hsub^{(2)} \, | \, \dots \, | \, \Hsub^{(\shots)} \right) \gets (\P\H)_{[t+1:n-k],[1:n]}\in\Fqm^{(n-t-k) \times n}$ \label{step:Hsub}
    \\
    \For{$i=1,\dots,\shots$}{
        Compute $\B^{(i)}\in\Fq^{t_i \times n_i}$ s.t. $\ext(\Hsub^{(i)})(\B^{(i)})^\top = \0$ and $\rk_q(\B^{(i)})=t_i$\label{step:getBi}
    }
    $\B \gets \diag(\B^{(1)}, \B^{(2)}, \dots, \B^{(\shots)})\in\Fq^{t \times n}$ \label{step:blockdiaB}
    \\
    Compute $\A\in\Fqm^{\intOrder \times t}$ s.t. $(\H\B^\top)\A^\top=\S$ \label{step:col_er_dec}
    \\
    $\C \gets \Y - \A\B\in\Fqm^{\intOrder \times n}$ \label{step:comp_E}
    \\
    \Return{$\C$}\label{step:return}
\end{algorithm}

\begin{theorem}
    Let $\C$ be a codeword of a homogeneous $\intOrder$-interleaved sum-rank-metric code $\mycode{IC}[\intOrder;\n,k,d]$ of minimum sum-rank distance $d$.
    Furthermore, let $\E \in \Fqm^{\intOrder \times n}$ be an error matrix of sum-rank weight $\SumRankWeight(\E)=t \leq d-2$ that fulfills $t\leq\intOrder$ (\emph{high-order condition})
    and $\rkqm(\E)=t$ (\emph{full-rank condition}). Then $\C$ can be uniquely recovered from the received word $\Y = \C + \E$ using Algorithm~\ref{alg:decode_high_order_int_sum_rank} in
    a time complexity equivalent to
    \begin{equation*}
        \oh{\max\{n^3, n^2\intOrder\}}
    \end{equation*}
    operations in $\Fqm$.
\end{theorem}

\begin{proof}
    By Lemma~\ref{lem:col_erasure_decoding} the error matrix $\E$ can be decomposed into $\E = \A\B$.
    Algorithm~\ref{alg:decode_high_order_int_sum_rank} first determines a basis of the error support $\SumRankSupp(\E)$ and then performs erasure decoding to obtain $\A$.
    The matrix $\B$ is computed by transforming $\S$ into row-echelon form using a transformation matrix $\P$ (see Line~\ref{step:transfP}).
    In Line~\ref{step:Hsub}, $\Hsub$ is obtained by choosing the last $n-k-t$ rows of $\P\H$.
    Then using Theorem~\ref{thm:sum_rank_support_recovery} for each block (see Line~\ref{step:getBi}) we find a matrix $\B^{(i)}$ whose rows form a basis for $\rker\left(\ext(\Hsub^{(i)})\right)_q$ and therefore a basis for $\RankSupp(\E^{(i)})$ for all $i\in[1:\shots]$.
    The matrix $\B$ is the block-diagonal matrix formed by $\B^{(i)}$ (cf.~\eqref{eq:def_B} and see Line~\ref{step:blockdiaB}).
    Finally, $\A$ can be computed from $\B$ and $\H$ using Lemma~\ref{lem:col_erasure_decoding} in Line~\ref{step:col_er_dec}.
    Hence, Algorithm~\ref{alg:decode_high_order_int_sum_rank} returns the transmitted codeword in Line~\ref{step:return}.
    The complexities of the lines in the algorithm are as follows:
    \begin{itemize}
        \item Line~\ref{step:syndromeMatrix}: The syndrome matrix $\S = \H \Y^\top$ can be computed in at most $\oh{n^2\intOrder}$ operations in $\Fqm$.
        \item Line~\ref{step:transfP}: The transformation of $[\S \,|\, \I]$ into row-echelon form requires \begin{equation*}\oh{(n-k)^2(\intOrder+n-k)}\subseteq\oh{\max\{n^3, n^2\intOrder\}}\end{equation*} operations in $\Fqm$.
        \item Line~\ref{step:Hsub}: The product $(\P\H)_{[t+1:n-k],[1:n]}$ can be computed requiring at most  $\oh{n(n-k-t)(n-k)}\subseteq\oh{n^3}$ operations in $\Fqm$.
        \item Line~\ref{step:getBi}: The transformation of $[\ext(\Hsub^{(i)})^\top \,|\, \I^\top]^\top$ into column-echelon form requires $\oh{n_i^2 ((n-k-t)m+n_i)}$ operations in $\Fq$ per block. Overall we get $\oh{\sum_{i=1}^{\shots} n_i^2 ((n-k-t)m+n_i)}\subseteq\oh{n^3 m}$ operations in $\Fq$ since we have that $\oh{\sum_{i=1}^{\shots} n_i^2}\subseteq\oh{n^2}$.
%        \item Line~\ref{step:blockdiaB}: \tj{Skip this?}
        \item Line~\ref{step:col_er_dec}: According to Lemma~\ref{lem:col_erasure_decoding}, this step can be done in $\oh{(n-k)^3 m^2}$ operations over $\Fq$.
%        \fh{Here's also an overful hbox in the previous line but it's at least not so bad.}
        \item Line~\ref{step:comp_E}: The product $\A\B = (\A^{(1)}\B^{(1)}\,|\,\A^{(2)}\B^{(2)}\,|\,\ldots\,|\,\A^{(\shots)}\B^{(\shots)})$ can be computed in $\sum_{i=1}^{\shots}\oh{\intOrder t_i n_i} \subseteq \oh{\intOrder n^2}$ and the difference of $\Y - \A\B$ can be computed in $\oh{\intOrder n}$ operations in $\Fqm$.
    \end{itemize}

    The complexities for {Line~\ref{step:getBi}} and Line~\ref{step:col_er_dec} are given for operations in $\Fq$.
    The number of $\Fq$-operations of both steps together is in $\oh{n^3 m^2}$ and their execution complexity can be bounded by $\oh{n^3}$ operations in $\Fqm$ (see~\cite{Couveignes2009}).
%    We prefer to give the time complexity in terms of $\Fqm$-operations even though some of the algorithm's operations have to be executed over $\Fq$.

%    \tj{Write here how the complexities are summarized and also how to approximate $\Fq \mapsto \Fqm$ complexity.}
%    We prefer to give the time complexity in terms of $\Fqm$-operations even though some of the algorithm's operations have to be executed over $\Fq$.
%    The number of $\Fq$-operations is in $\oh{n^3 m^2}$ and their execution complexity can be bounded by $\oh{n^3}$ operations in $\Fqm$ (see~\cite{Couveignes2009}).

    Thus, Algorithm~\ref{alg:decode_high_order_int_sum_rank} requires $\oh{\max\{n^3, n^2\intOrder\}}$ operations in $\Fqm$ and $\oh{n^3 m^2}$ operations in $\Fq$.
\qed
\end{proof}

%\tj{Write here also about the complexity of generic}
Note that the complexity of Algorithm~\ref{alg:decode_high_order_int_sum_rank} is not affected by the decoding complexity of the underlying constituent code since a generic code with no structure is assumed.
%\fh{I like that sentence. Thumbs up. :-)}

\begin{example}
    Let $\Fqm = \F_{5^2}$ with the primitive polynomial $x^2+4x+2$ and the primitive element $\alpha$ be given.
    Further let $\mycode{IC}[\intOrder;\n,k,d]$ be an interleaved sum-rank-metric code of length $n$ with $\n = (2,2,2)$, $k=2$, $d=5$, $\shots = 3$ and $\intOrder = 3$, defined by a generator matrix
    \begin{equation*}
        \G = \left(\begin{array}{cc|cc|cc}
                 \alpha^{4}  &  \alpha^{7}  & \alpha^{21} & \alpha^{4}  & \alpha^{3}  & \alpha^{5}  \\
                 \alpha^{20} &  \alpha^{11} & \alpha^{10} & \alpha^{21} & \alpha^{17} & \alpha^{3}
        \end{array}\right)
    \end{equation*}
    and a parity-check matrix
    \begin{equation*}
        \H = \left(
        \newcolumntype{C}[1]{>{\centering\arraybackslash}p{#1}}
        \begin{array}{C{13pt}C{13pt}|C{13pt}C{13pt}|cc}
                 {1}  &  {0}  & {0} & {0}  & \alpha^{8}  & \alpha^{19}  \\
                 {0}  &  {1}  & {0} & {0}  & \alpha^{5}  & \alpha^{12}  \\
                 {0}  &  {0}  & {1} & {0}  & \alpha^{17}  & \alpha^{}  \\
                 {0}  &  {0}  & {0} & {1}  & \alpha^{22}  & \alpha^{18}
        \end{array}\right).
    \end{equation*}
%    \fh{We could write $\alpha$ instead of $\alpha^1$ in the last column.}
    Suppose that the codeword
    \begin{equation*}
        \C = \left(\begin{array}{cc|cc|cc}
                 \alpha^{20} &  \alpha^{22} & 1           & \alpha^{6} & \alpha^{11} & \alpha^{10}  \\
                 \alpha^{23} &  \alpha^{7} & \alpha^{4} & 0          & \alpha^{17} & \alpha^{9}  \\
                 \alpha^{15} &  1          & \alpha^{22} & \alpha^{12} & \alpha^{22} & \alpha^{10}
        \end{array}\right) \in\mycode{IC}[\intOrder;\n,k,d]
    \end{equation*}
    is corrupted by the error
    \begin{equation*}
        \E = \left(
        \newcolumntype{C}[1]{>{\centering\arraybackslash}p{#1}}
        \begin{array}{cc|cc|C{13pt}C{13pt}}
                       \alpha^{19} &  \alpha^{} & \alpha^{6}  & \alpha^{9} & 0 & 0  \\
                       \alpha^{17} &  \alpha^{23} & \alpha^{10} & \alpha^{7}   & 0 & 0  \\
                       \alpha^{2} &  \alpha^{8} & \alpha^{15} & \alpha^{6} & 0 & 0
        \end{array}\right)
    \end{equation*}
%    \fh{We could write $\alpha$ instead of $\alpha^1$ in the first row.}
    with $\SumRankWeight(\E)= \rk_{q^m}(\E) =  t = 3$ and $t_1 = 1, t_2=2$ and $t_3 = 0$.
    The resulting received matrix is
    \begin{equation*}
        \Y = \left(\begin{array}{cc|cc|cc}
                       \alpha^{17} &  \alpha^{8} & \alpha^{18} & \alpha^{16} & \alpha^{11} & \alpha^{10}  \\
                       \alpha^{11} &  \alpha^{3} & \alpha^{22} & \alpha^{7} & \alpha^{17} & \alpha^{9}  \\
                       \alpha^{7} &  \alpha^{4}  & \alpha^{23} & 1 & \alpha^{22} & \alpha^{10}
        \end{array}\right).
    \end{equation*}
    First the syndrome matrix is computed as
    \begin{equation*}
        \S = \H\Y^{\top} = \mleft(\begin{array}{ccc}
                       \alpha^{19} &  \alpha^{17} & \alpha^{2} \\
                       \alpha^{} &  \alpha^{23} & \alpha^{8} \\
                       \alpha^{6} &  \alpha^{10}  & \alpha^{15} \\
                       \alpha^{9} &  \alpha^{7}  & \alpha^{6}
        \end{array}\mright)
    \end{equation*}
%    \fh{We could write $\alpha$ instead of $\alpha^1$ in the first column.}
    and then $\P$
    \begin{equation*}
        \P = \left(
        \newcolumntype{C}[1]{>{\centering\arraybackslash}p{#1}}
        \begin{array}{C{13pt}ccc}
                        0 & \alpha^{2} & \alpha^{6} & \alpha^{8} \\
                        0 & \alpha^{4} & \alpha^{20} & \alpha^{15} \\
                        0 & \alpha^{23} & 0 & \alpha^{3} \\
                        1 & \alpha^{6} & 0 & 0 \\
        \end{array} \right) \Longrightarrow \P\S = \left(
        \newcolumntype{C}[1]{>{\centering\arraybackslash}p{#1}}
        \begin{array}{C{13pt}C{13pt}C{13pt}C{13pt}}
                       1 & 0 & 0 \\
                       0 & 1 & 0 \\
                       0 & 0 & 1 \\
                       0 & 0 & 0 \\
        \end{array}\right)
    \end{equation*}
    with $\rk_{q^m}(\P)=4$.
    The last $n-k-t = 1$ rows of
    \begin{equation*}
        \P\H = \left(
        \newcolumntype{C}[1]{>{\centering\arraybackslash}p{#1}}
        \begin{array}{C{13pt}c|cc|cc}
            0 & \alpha^{2} & \alpha^{6} & \alpha^{8} & \alpha^{13} & \alpha^{7}  \\
            0 & \alpha^{4} & \alpha^{20} & \alpha^{15} & \alpha^{22} & \alpha^{16} \\
            0 & \alpha^{23} & 0 & \alpha^{3} & \alpha^{11} & 1 \\
            1 & \alpha^{6} & 0 & 0 & \alpha^{18} & \alpha^{16}\\
        \end{array} \right)
    \end{equation*}
    gives us $\Hsub = (1 \,\,\, \alpha^{6}\,|\,0\,\,\,0\,|\,\alpha^{18}\,\,\,\alpha^{16})$.
    We expand every block of $\Hsub$ over $\F_{5}$ and get
    \begin{equation*}
        \ext(\Hsub^{(1)}) = \begin{pmatrix}
        1 & 2 \\
        0 & 0
        \end{pmatrix}, \quad
        \ext(\Hsub^{(2)}) = \begin{pmatrix}
                   0 & 0 \\
                   0 & 0
        \end{pmatrix} \quad \text{and} \quad
        \ext(\Hsub^{(3)}) = \begin{pmatrix}
                   3 & 3 \\
                   0 & 3
        \end{pmatrix}.
    \end{equation*}
    We observe that the second block $\Hsub^{(2)}$ is zero which corresponds to a full-rank error.
    Next we compute a basis for each of the right kernels of $\ext(\Hsub^{(1)})$, $\ext(\Hsub^{(2)})$, and $\ext(\Hsub^{(3)})$ which gives us
    \begin{equation*}
        \B^{(1)} = \begin{pmatrix}
                       1 & 2
        \end{pmatrix}, \quad
        \B^{(2)} = \begin{pmatrix}
                       1 & 0 \\
                       0 & 1
        \end{pmatrix}, \quad
        \B^{(3)} = (),
    \end{equation*}
    where $\B^{(3)}$ is empty since $\ext(\Hsub^{(3)})$ has full rank.
    The matrix $\B$ is then given by
    \begin{equation*}
        \B = \diag(\B^{(1)},\B^{(2)},\B^{(3)})
           = \mleft(\begin{array}{cc|cc|cc}
                 1 & 2 & 0 & 0 & 0 & 0 \\\hline
                 0 & 0 & 1 & 0 & 0 & 0 \\
                 0 & 0 & 0 & 1 & 0 & 0 \\\hline
        \end{array}\mright).
    \end{equation*}
    Solving for $\A$
    \begin{align*}
        \H\B^\top\A^\top &= \S \\
        \begin{pmatrix}
            1 & 0 & 0 \\
            \alpha^{6} & 0 & 0 \\
            0 & 1 & 0 \\
            0 & 0 & 1
        \end{pmatrix}\A^\top &= \mleft(\begin{array}{ccc}
                       \alpha^{19} &  \alpha^{17} & \alpha^{2} \\
                       \alpha^{} &  \alpha^{23} & \alpha^{8} \\
                       \alpha^{6} &  \alpha^{10}  & \alpha^{15} \\
                       \alpha^{9} &  \alpha^{7}  & \alpha^{6}
        \end{array}\mright)
    \end{align*}
%    \fh{We could write $\alpha$ instead of $\alpha^1$ in the first column of the right-hand side.}
    gives
    \begin{equation*}
    \A^\top = \begin{pmatrix}
%                    \alpha^{19} & \alpha^{17} & \alpha^{2} \\
                    \alpha^{19} & \alpha^{17} & \alpha^{2} \\
                    \alpha^{6} & \alpha^{10} & \alpha^{15} \\
                    \alpha^{9} & \alpha^{7} & \alpha^{6}
    \end{pmatrix}\Rightarrow \hat{\E} = \A\B =  \left(
        \newcolumntype{C}[1]{>{\centering\arraybackslash}p{#1}}
        \begin{array}{cc|cc|C{13pt}C{13pt}}
                       \alpha^{19} &  \alpha^{} & \alpha^{6}  & \alpha^{9} & 0 & 0  \\
                       \alpha^{17} &  \alpha^{23} & \alpha^{10} & \alpha^{7}   & 0 & 0  \\
                       \alpha^{2} &  \alpha^{8} & \alpha^{15} & \alpha^{6} & 0 & 0
        \end{array}\right)
    \end{equation*}
    and $\hat{\E} = \E$. Finally, the codeword $\C$ can be recovered as $\C = \Y - \hat{\E}$.
\end{example}

\section{Implications for Decoding High-Order Interleaved Skew-Metric Codes}

The \emph{skew metric} is closely related to the sum-rank metric and was first considered in~\cite{martinez2018skew}.
In particular, there exists an isometry between the sum-rank metric and the skew metric for most code parameters (see~\cite[Theorem~3]{martinez2018skew}).

%We show in this section how a skew-metric code can be constructed from a high-order interleaved sum-rank-metric code.
We show in this section how an interleaved skew-metric code can be constructed from a high-order interleaved sum-rank-metric code.
%\fh{Should this be either "how an interleaved skew-metric code ... from a high-order interleaved sum-rank-metric code" or "how a skew-metric code ... from a sum-rank-metric code"? Or do we maybe rather want to write something like "how to transfer a high-order interleaved sum-rank-metric code into the skew metric"?}
%\tj{Good question! I'm not sure. But as I understand both should be fine but as it is it is mixed interleaved and non-interleaved. Maybe Hannes can comment on this.}
This enables us to apply the presented decoder to the obtained high-order interleaved skew-metric codes and correct
errors of a fixed skew weight.

The mentioned isometry can be described and applied to the interleaved context as follows:
Let us consider vectors from $\Fqm^n$, where $n$ satisfies the constraints in~\cite[Theorem~2]{martinez2018skew}.
By~\cite[Theorem~3]{martinez2018skew}, there exists an invertible diagonal matrix $\D \in \Fqm^{n \times n}$ such that
%\fh{To me, the double parentheses are not that well readable. Maybe we can write "such that <equation> holds for the skew weight (see ... for its definition)".}
%\tj{I agree. I made some changes which I think are better now}
\begin{equation}\label{eq:def_isometry}
    \SumRankWeight(\x\D) = \SkewWeight(\x), \qquad \forall \x \in \Fqm^n,
\end{equation}
where for the definition of the skew weight $\SkewWeight(\cdot)$ see~\cite[Definition~9]{martinez2018skew}.
The skew metric for interleaved matrices has been considered in~\cite{bartz2022fast}.
Namely, the extension of~\eqref{eq:def_isometry} to $\Fqm^{\intOrder \times n}$, we get (see~\cite{bartz2022fast})
%\fh{I personally don't like that we cite the same thing twice within two sentences. Maybe change the second sentence to "Namely, the extension of (36) to $\Fqm^{\intOrder \times n}$ yields"?}
\begin{equation}\label{eq:def_isometry_int}
    \SumRankWeight(\X\D) = \SkewWeight(\X), \qquad \forall \X \in \Fqm^{\intOrder \times n}.
\end{equation}
Now consider a linear $\intOrder$-interleaved sum-rank-metric code $\mycode{IC}[\intOrder;\n,k,\dminsr]$ with parity-check matrix $\H$.
%\fh{In the reasoning below, it could be helpful to write $\mycode{IC}_{\Sigma R}[\intOrder;\n,k,\dminsr]$ similar to $\mycode{IC}_{skew}[\intOrder; n, k, d]$ for clarity.}
Then by~\eqref{eq:def_isometry_int} the code
\begin{equation*}
    \mycode{IC}_{skew}[\intOrder; n, k, d] \defeq 
    \left\{\C\D^{-1}:\C \in \mycode{IC}[\intOrder; \n, k, d] \right\}
\end{equation*}
is an $\intOrder$-interleaved skew-metric code with minimum skew distance $d$.
Observe that the parity-check matrix of the constituent skew-metric code $\mycode{C}_{skew}[n,k,d]$ of $\mycode{IC}_{skew}[\intOrder; n, k, d]$ is given by $\H_{skew} = \H\D$.

Let us now study a decoding problem related to the obtained skew-metric code.
Consider a matrix $\Y = \C + \E$ where $\C \in \mycode{IC}_{skew}[\intOrder; n, k, d]$ and $\E$ is an error matrix with $\SkewWeight(\E)=t$.
Then~\eqref{eq:def_isometry_int} implies that we have
\begin{equation*}
    \widetilde{\Y} \defeq (\C + \E)\D = \widetilde{\C} + \widetilde{\E}
\end{equation*}
where $\widetilde{\C} \in \mycode{IC}[\intOrder; \n, k, d]$ and $\SumRankWeight(\E)=t$.
Hence, using the isometry from~\cite[Theorem~3]{martinez2018skew} we can map the decoding problem in the skew metric to the sum-rank metric and vice versa (see also~\cite{bartz2022fast}).

In particular, this allows us to use Algorithm~\ref{alg:decode_high_order_int_sum_rank} to solve the posed decoding problem in the skew metric.
The steps to decode a high-order interleaved skew-metric code $\mycode{IC}_{skew}[\intOrder; n, k, d]$ with parity-check matrix $\H_{skew}$ (whose parameters comply with~\cite[Theorem~2]{martinez2018skew}) can be summarized as follows:
\begin{enumerate}
    \item Compute the transformed received matrix $\widetilde{\Y} \defeq (\C + \E)\D=\widetilde{\C} + \widetilde{\E}$ where $\C \in \mycode{IC}_{skew}[\intOrder; n, k, d]$ and $\SkewWeight(\E)=t$.

    \item Apply Algorithm~\ref{alg:decode_high_order_int_sum_rank} to $\widetilde{\Y}$. If $\rk_{q^m}(\widetilde{\E})=t$, which is equivalent to $\rk_{q^m}(\E)=t$, the algorithm recovers $\widetilde{\C} \in \mycode{IC}[\intOrder; \n, k, d]$.

    \item Recover $\C \in \mycode{IC}_{skew}[\intOrder; n, k, d]$ as $\C = \widetilde{\C}\D^{-1}$.
\end{enumerate}

Since the first and the the third step both require $\oh{\intOrder n}$ operations in $\Fqm$, the overall complexity is
dominated by the complexity of Algorithm~\ref{alg:decode_high_order_int_sum_rank}, that is
$\oh{\max\{n^3, n^2\intOrder\}}$ operations in $\Fqm$.

%----------------------------------------------------------------------------
% Comparison
%----------------------------------------------------------------------------
\section{Comparison of Metzner-Kapturowski-like Decoders in the Hamming, Rank and Sum-Rank Metric} \label{sec:comparison}
%\begin{itemize}
%    \item Understand in detail how the decoders for Hamming and rank metric are recovered
%    \begin{itemize}
%        \item \todo{See if we get something new for the Hamming metric (like e.g. a Metzner--Kapturowski decoder for \emph{generalized} \ac{RS} codes) $\Rightarrow$ FH}
%
%        \item Here we might get a nice figure
%    \end{itemize}
%
%    \item Show relation to known \acs{ILRS} decoders
%\end{itemize}

The decoder presented in Algorithm~\ref{alg:decode_high_order_int_sum_rank} is a generalization of the Metzner--Kaptu\-row\-ski decoder for the Hamming metric~\cite{metzner1990general} and the Metzner--Kapturowski-like decoder for the rank metric~\cite{puchinger2019decoding}.
In this section we illustrate how the proposed decoder works in three different metrics: 1.) Hamming metric, 2.) Rank metric and 3.) Sum-rank metric. 
Note that the Hamming and the rank metric are both special cases of the sum-rank metric.
We also show the analogy of the different definitions of the error support for all three cases.

The support for the Hamming-metric case is defined as
\begin{equation*}
    \HammingSupp(\E) \defeq \{j \,:\, j\text{-th column of } \E \text{ is non-zero}\}.
\end{equation*}
%\fh{Here, we use $|$ to separate the condition in a set definition. However, we sometimes also use $:$ in the paper (e.g. in the definition of interleaved codes). Can we make this consistent?}
In the Hamming metric an error matrix $\E$ with $t_{H}$
%\fh{Why don't we simply use $t$ here? I don't think it would be a naming conflict.}\tj{I thought it might be worth in the comparison section to distinguish and point out that the weights are in different metrics.} \fh{I understand the intention. It's certainly not wrong even though I think it's not necessary. But I'm also happy if we keep it that way. :-)}
errors can be decomposed into $\E=\A\B$, where the rows of $\B$ are the unit vectors corresponding to the $t_H$ error positions.
This means the support of the error matrix is given by the union of the supports of the rows $\B_i$ of $\B$ ($\forall i\in[1:t_H]$), hence
\begin{equation*}
    \HammingSupp(\E) = \bigcup\limits_{i=1}^{t_H}\HammingSupp(\B_i).
\end{equation*}
If the condition for the Metzner--Kapturowski decoder is fulfilled (\emph{full-rank condition}), then the zero columns in $\Hsub$ indicate the error positions and thus give rise to the error support, i.e. we have that
\begin{equation*}
    \HammingSupp(\E) = [1:n] \setminus \bigcup\limits_{i=1}^{n-k-t_H} \HammingSupp(\H_{\text{sub},i})
\end{equation*}
where $\H_{\text{sub},i}$ is the $i$-th row of $\Hsub$.
%\fh{Maybe state that $\H_{\text{sub},i}$ stands for the $i$-th row of $\H_{sub}$ (similar to $\B_i$ before).}
Figure~\ref{fig:hamming_figure} illustrates how the error support $\HammingSupp(\E)$ can be recovered from $\Hsub$.

The rank-metric case is similar, except for a different notion
%\fh{Do we mean "notion" here? It's not only notation...}\tj{Yeah, you are right. Or maybe even "definition"?} \fh{I think both solutions are fine. I'd personally use "notion" but I simply like the word. :D}
for the error support.
Again, the error $\E$ with $\rkq(\E) = t_R$ can be decomposed as $\E=\A\B$. Then the rank support $\RankSupp(\E)$ of $\E$ equals the row space of $\ext{(\B)}$, which is spanned by the union of all rows of $\ext{(\B_i)}$ with $\B_i$ being the $i$-th row of $\B$.
This means the support of $\E$ is given by
\begin{equation*}
    \RankSupp(\E) = \mksum\limits_{i=1}^{t_R} \RankSupp(\B_i)
\end{equation*}
%\fh{We have never defined $t_R$ (maybe add it when the decomposition is explained). In my opinion, we could still use $t$ though.}
with $\mksum$ being the addition of vector spaces, which means the span of the union of the considered spaces.
If the condition on the error matrix (\emph{full-rank condition}) is fulfilled, the rank support of $\E$ is given by the kernel of $\ext(\Hsub)$~\cite{renner2021decoding}.
As illustrated in Figure~\ref{fig:rank_figure} the row space of $\ext(\Hsub)$ can be computed by obtaining the span of the union of spaces $\RankSupp(\H_{\text{sub},i})$, where $\H_{\text{sub},i}$ is the $i$-th row of $\Hsub$.
Finally, the support of $\E$ is given by
\begin{equation*}
    \RankSupp(\E) = \left(\mksum\limits_{i=1}^{n-k-t_R}\RankSupp(\H_{\text{sub},i})\right)^{\perp}.
\end{equation*}

For the sum-rank metric we get from~\eqref{eq:sumranksupport} that
\begin{align*}
    \SumRankSupp(\E) &= \RankSupp{(\B^{(1)})} \times \RankSupp{(\B^{(2)})} \times \dots \times \RankSupp{(\B^{(\shots)})} \\
    &= \left(\mksum\limits_{i=1}^{n-k-t_1} \RankSupp(\B_{1}^{(1)})\right) \times\dots\times\left(\mksum\limits_{i=1}^{n-k-t_\shots} \RankSupp(\B_{\shots}^{(\shots)})\right).
\end{align*}
%\fh{Maybe state once more what $t_1, \dots, t_{\shots}$ are. It seems quite a while ago since we used them last.}

According to Theorem~\ref{thm:sum_rank_support_recovery} we have that
\begin{align*}
    \SumRankSupp(\E) &= \left(\mksum\limits_{i=1}^{n-k-t_1} \RankSupp(\H_{\text{sub},1}^{(1)})\right)^{\perp} \times\dots\\
    &\dots\times\left(\mksum\limits_{i=1}^{n-k-t_\shots} \RankSupp(\H_{\text{sub},\shots}^{(\shots)})\right)^{\perp}.
\end{align*}
%\tj{Maybe add a table to compare the support notations of the three cases.}
%\fh{overfull hbox in the previous line.}
%\fh{Add a sentence that refers to Figure 3.}
The relation between the error matrix $\E$, the matrix $\Hsub$ and the error supports for the Hamming metric, rank metric and sum-rank metric are illustrated in Figure~\ref{fig:hamming_figure}, Figure~\ref{fig:rank_figure} and Figure~\ref{fig:sum_rank_figure}, respectively.

%\begin{align}\label{eq:sumranksupport}
%\SumRankSupp(\E)
%\defeq
%&\Rowspace{\E^{(1)}} \times \Rowspace{\E^{(2)}} \times \dots \times \Rowspace{\E^{(\shots)}}
%\\
%=&\Rowspace{\B^{(1)}} \times \Rowspace{\B^{(2)}} \times \dots \times \Rowspace{\B^{(\shots)}}.
%\end{align}

\begin{figure}[ht!]
    \centering
    \begin{tikzpicture}[scale=.45]

    \draw[] (0,0) rectangle (0.5,3);
    \draw[] (0.5,0) rectangle (1,3);
    \draw[] (1,0) rectangle (1.5,3);
    \draw[] (1.5,0) rectangle (2,3);
    \draw[pattern=north west lines, pattern color=black] (2,0) rectangle (2.5,3);
    \draw[pattern=north west lines, pattern color=black] (2.5,0) rectangle (3,3);
    \draw[] (3,0) rectangle (3.5,3);
    \draw[] (3.5,0) rectangle (4,3);
    \draw[] (4,0) rectangle (4.5,3);
    \draw[] (4.5,0) rectangle (5,3);
    \draw[] (5,0) rectangle (5.5,3);
    \draw[] (5.5,0) rectangle (6,3);
    \draw[pattern=north west lines, pattern color=black] (6,0) rectangle (6.5,3);
    \draw[] (6.5,0) rectangle (7,3);
    \draw[pattern=north west lines, pattern color=black] (7,0) rectangle (7.5,3);
    \draw[] (7.5,0) rectangle (8,3);
    \draw[] (8,0) rectangle (8.5,3);
    \draw[] (8.5,0) rectangle (9,3);

    \draw[-latex] (2.2,-0.5) -- (2.2,0);
    \draw[-latex] (2.7,-0.5) -- (2.7,0);
    \draw[-latex] (6.2,-0.5) -- (6.2,0);
    \draw[-latex] (7.2,-0.5) -- (7.2,0);
    \node[] () at (5.0, -1.0) {error positions}; %each block $\in Fq^{s\times n_i}$ and of $\Fq$-rank $t_i$};

    \node[] () at (-1, 1.5) {$\E = $};
%    \node[] () at (10, 1.5) {$s$};
%    \node[] () at (4, 3.5) {$n_i$};

    \node[] () at (10, 1.5) {$=$};

    \draw[pattern=north west lines, pattern color=black] (11,0) rectangle (14,4);
%    \draw[pattern=north west lines, pattern color=black] (12,0) rectangle (12.5,4);
%    \draw[pattern=north west lines, pattern color=black] (12.5,0) rectangle (14,4);
%    \node[] () at (12.5, 4.5) {$t$};
    \node[] () at (12.5, -0.5) {$\A$};
    \node[] () at (14.5, 1.5) {$\cdot$};
%    \draw[] (15,0.5) rectangle (24,3.5);

    \draw[] (15,0) rectangle (15.5,3);
    \draw[] (15.5,0) rectangle (16,3);
    \draw[] (16,0) rectangle (16.5,3);
    \draw[] (16.5,0) rectangle (17,3);
    \draw[] (17,0) rectangle (17.5,3);
    \draw[] (17.5,0) rectangle (18,3);
    \draw[] (18,0) rectangle (18.5,3);
    \draw[] (18.5,0) rectangle (19,3);
    \draw[] (19,0) rectangle (19.5,3);
    \draw[] (19.5,0) rectangle (20,3);
    \draw[] (20,0) rectangle (20.5,3);
    \draw[] (20.5,0) rectangle (21,3);
    \draw[] (21,0) rectangle (21.5,3);
    \draw[] (21.5,0) rectangle (22,3);
    \draw[] (22,0) rectangle (22.5,3);
    \draw[] (22.5,0) rectangle (23,3);
%    \node[draw,rectangle,pattern=north west lines, pattern color=gray, minimum width=0.5cm, minimum height=0.75cm] () at (22.5, 0.70) {};
    \draw[] (22,0) rectangle (22.5,3);
    \draw[] (22.5,0) rectangle (23,3);
    \draw[] (23,0) rectangle (23.5,3);
    \draw[] (23.5,0) rectangle (24,3);
    \node[] () at (19, -0.5) {$\B$};
%    \node[] () at (19, 4.0) {$n$};
%    \node[] () at (24.5, 1.5) {$t$};
       \node[] () at (17.3, 2.6) {\small{1}};
    \node[] () at (17.8, 1.9) {\small{1}};
    \node[] () at (21.3, 1.2) {\small{1}};
    \node[] () at (22.3, 0.4) {\small{1}};

    \draw[-latex] (17.2,-1.0) -- (17.2,0);
    \draw[-latex] (17.7,-1.0) -- (17.7,0);
    \draw[-latex] (21.2,-1.0) -- (21.2,0);
    \draw[-latex] (22.2,-1.0) -- (22.2,0);

    \node[] () at (19.5, -1.5) {error positions};

    \node[] () at (-1.7, -5) {$\bm{H}_\text{sub}=$};
    \draw[pattern=north west lines] (0, -3.5) rectangle (0.5, -6.5);
    \draw[pattern=north west lines] (0.5, -3.5) rectangle (1, -6.5);
    \draw[pattern=north west lines] (1, -3.5) rectangle (1.5, -6.5);
    \draw[pattern=north west lines] (1.5, -3.5) rectangle (2, -6.5);
    \draw[] (2, -3.5) rectangle (2.5, -6.5);
    \draw[] (2.5, -3.5) rectangle (3, -6.5);
    \draw[pattern=north west lines] (3, -3.5) rectangle (3.5, -6.5);
    \draw[pattern=north west lines] (3.5, -3.5) rectangle (4, -6.5);
    \draw[pattern=north west lines] (4, -3.5) rectangle (4.5, -6.5);
    \draw[pattern=north west lines] (4.5, -3.5) rectangle (5, -6.5);
    \draw[pattern=north west lines] (5, -3.5) rectangle (5.5, -6.5);
    \draw[pattern=north west lines] (5.5, -3.5) rectangle (6, -6.5);
    \draw[] (6, -3.5) rectangle (6.5, -6.5);
    \draw[pattern=north west lines] (6.5, -3.5) rectangle (7, -6.5);
    \draw[] (7, -3.5) rectangle (7.5, -6.5);
    \draw[pattern=north west lines] (7.5, -3.5) rectangle (8, -6.5);
    \draw[pattern=north west lines] (8, -3.5) rectangle (8.5, -6.5);
    \draw[pattern=north west lines] (8.5, -3.5) rectangle (9, -6.5);
%    \draw[] (4, -3.5) rectangle (5, -6.5);
%    \draw[] (5, -3.5) rectangle (7, -6.5);
%    \draw[] (6, -3.5) rectangle (7, -6.5);
%    \node[] () at (11, -5) {$n-k-t$};
%    \node[] () at (4, -7.1) {$n$};

    \draw[-latex] (2.2,-7.5) -- (2.2,-6.5);
    \draw[-latex] (2.7,-7.5) -- (2.7,-6.5);
    \draw[-latex] (6.2,-7.5) -- (6.2,-6.5);
    \draw[-latex] (7.2,-7.5) -- (7.2,-6.5);
%    \draw[-latex] (7.5,-7.5) -- (7.5,-6.5);
    \node[] () at (5.0, -8.0) {all-zero columns in error positions};
    \node[anchor=west] () at (10.0, -5) {$\Rightarrow \HammingSupp(\E) =\bigcup\limits_{i=1}^{t_H}\HammingSupp(\B_i) =  $};
    \node[anchor=west] () at (12, -7.5) {$ =[1:n] \setminus \bigcup\limits_{i=1}^{n-k-t_H}\HammingSupp(\H_{\text{sub},i})$};
\end{tikzpicture}
    \caption{Illustration of the error support for the Hamming-metric case with $\E\in\Fqm^{\intOrder \times n}$, $\A\in\Fqm^{\intOrder\times t_H}$, $\B\in\Fq^{t_H\times n}$ and $\Hsub\in\Fqm^{(n-k-t_H)\times n}$. $\B_i$ is the $i$-th row of $\B$ and $\H_{\text{sub},i}$ the $i$-th row of $\Hsub$.}\label{fig:hamming_figure}
\end{figure}

\begin{figure}[ht!]
    \centering
    \begin{tikzpicture}[scale=.45]

    \draw[pattern=north west lines, pattern color=black] (0,0) rectangle (9,3);

    \node[] () at (-1, 1.5) {$\E = $};
%    \node[] () at (10, 1.5) {$s$};
%    \node[] () at (4, 3.5) {$n_i$};

    \node[] () at (10, 1.5) {$=$};
    \node[] () at (19, 1.5) {$\Fq$};

    \draw[pattern=north west lines, pattern color=black] (11,0) rectangle (14,4);
%    \node[] () at (12.5, 4.5) {$t$};
    \node[] () at (12.5, -0.5) {$\A$};
    \node[] () at (14.5, 1.5) {$\cdot$};
%    \draw[] (15,0.5) rectangle (24,3.5);
    \draw[] (15,0) rectangle (24,3);
    \node[] () at (19, -0.5) {$\B$};
%    \node[] () at (19, 4.0) {$n$};
%    \node[] () at (24.5, 1.5) {$t$};

%
%    \node[] (T1) at (16.5, 4.32) {$\in \Fq^{t_i \times n_i}$};
%    \draw[-latex] (T1) -- (T2);
%
    \node[] () at (0, -5) {$\bm{H}_\text{sub} = $};

    \draw[] (3, -3.0) rectangle node {$\H_{\text{sub},1}$} (8, -4.5);
    \draw[] (3, -4.5) rectangle node {$\vdots$} (8, -6.0);
    \draw[] (3, -6.0) rectangle node {$\H_{\text{sub},n-k-t}$} (8, -7.5);

        \node[] () at (11, -5) {$\mapsto \ext(\bm{H}_\text{sub}) = $};

    \draw[] (14, -2.0) rectangle node {\makecell{generating set of \\ $\RankSupp(\H_{\text{sub},1})$}} (23.5, -4.5);
    \draw[] (14, -4.5) rectangle node {$\vdots$} (23.5, -6.0);
    \draw[] (14, -6.0) rectangle node {\makecell{generating set of \\ $\RankSupp(\H_{\text{sub},n-k-t})$}} (23.5, -8.5);

\end{tikzpicture}
    \caption{Illustration of the error support for the rank-metric case with $\E\in\Fqm^{\intOrder \times n}$, $\A\in\Fqm^{\intOrder\times t_R}$, $\B\in\Fq^{t_R\times n}$ and $\Hsub\in\Fqm^{(n-k-t_R)\times n}$. $\B_i$ is the $i$-th row of $\B$ and $\H_{\text{sub},i}$ the $i$-th row of $\Hsub$.}\label{fig:rank_figure}
\end{figure}

\begin{figure}[ht!]
    \centering
    \begin{tikzpicture}[scale=.45]

    \draw[] (0,0) rectangle (1,3);
    \draw[] (1,0) rectangle (2,3);
    \draw[pattern=north west lines, pattern color=black] (2,0) rectangle (3,3);
    \draw[] (3,0) rectangle (5,3);
    \draw[pattern=north west lines, pattern color=black] (5,0) rectangle (7,3);
    \draw[pattern=north west lines, pattern color=black] (7,0) rectangle (8,3);
    \draw[] (7,0) rectangle (9,3);

    \draw[-latex] (2.5,-0.5) -- (2.5,0);
    \draw[-latex] (6.0,-0.5) -- (6.0,0);
    \draw[-latex] (7.5,-0.5) -- (7.5,0);
    \node[] () at (5.0, -1.0) {blocks with rank errors}; %each block $\in Fq^{s\times n_i}$ and of $\Fq$-rank $t_i$};

    \node () at (0.90, 4) {$\E^{(1)}$};
    \draw[-latex] (0.60,3.5) -- (0.60,3);
    \node () at (4.90, 4) {$\cdots$};
    \node () at (8.90, 4) {$\E^{(\shots)}$};
    \draw[-latex] (8.60,3.5) -- (8.60,3);

    \node[] () at (-1, 1.5) {$\E = $};
%    \node[] () at (10, 1.5) {$s$};
%    \node[] () at (4, 3.5) {$n_i$};

    \node[] () at (10, 1.5) {$=$};

    \draw[pattern=north west lines, pattern color=black] (11,0) rectangle (12,4);
    \draw[pattern=north west lines, pattern color=black] (12,0) rectangle (12.5,4);
    \draw[pattern=north west lines, pattern color=black] (12.5,0) rectangle (14,4);
%    \node[] () at (12.5, 4.5) {$t$};
    \node[] () at (12.5, -0.5) {$\A$};
    \node[] () at (14.5, 1.5) {$\cdot$};
%    \draw[] (15,0.5) rectangle (24,3.5);
    \draw[] (15,0) rectangle (16,3);
    \draw[] (16,0) rectangle (17,3);
    \draw[] (17,0) rectangle (18,3);
    \draw[pattern=north west lines, pattern color=black] (17, 2) rectangle (18, 3);
    \draw[] (18,0) rectangle (20,3);
    \draw[] (20,0) rectangle (22,3);
    \draw[pattern=north west lines, pattern color=black] (20, 1.5) rectangle (22, 2);
    \draw[] (22,0) rectangle (23,3);
%    \node[draw,rectangle,pattern=north west lines, pattern color=gray, minimum width=0.5cm, minimum height=0.75cm] () at (22.5, 0.70) {};
    \draw[pattern=north west lines, pattern color=black] (22, 0) rectangle (23, 1.5);
    \draw[] (22,0) rectangle (24,3);
    \node[] () at (19, -0.5) {$\B$};
%    \node[] () at (19, 4.0) {$n$};
%    \node[] () at (24.5, 1.5) {$t$};

    \node () at (11.5, 5.0) {$\A^{(1)}$};
    \draw[-latex] (11.20, 4.5) -- (11.20, 4.0);
    \node () at (12.80, 5.0) {$\cdots$};
    \node () at (14.0, 5.0) {$\A^{(\shots)}$};
    \draw[-latex] (13.50, 4.5) -- (13.50, 4.0);
%    \draw[-latex] (8.60,3.5) -- (8.60,3);

    \node () at (17.90, 4) {$\B^{(i)}$};
    \draw[-latex] (17.50,3.5) -- (17.50,3);
%    \node () at (15.90, 4) {$\B^{(1)}$};
%    \draw[-latex] (15.60,3.5) -- (15.60,3);
%    \node () at (19.90, 4) {$\cdots$};
%    \node () at (23.90, 4) {$\B^{(\shots)}$};
%    \draw[-latex] (23.60,3.5) -- (23.60,3);
%
    \draw[-latex] (17.5, -1.0) -- (17.5, 0);
    \draw[-latex] (21, -1.0) -- (21,0);
    \draw[-latex] (22.5,-1.0) -- (22.5,0);

    \node[] () at (19.5, -1.5) {blocks with rank errors};

    \node[] () at (-1.7, -5) {$\bm{H}_\text{sub}=$};
    \draw[pattern=north west lines] (0, -3.5) rectangle (1, -6.5);
    \draw[pattern=north west lines] (1, -3.5) rectangle (2, -6.5);
    \draw[] (2, -3.5) rectangle (3, -6.5);
    \draw[pattern=north west lines] (3, -3.5) rectangle (5, -6.5);
%    \draw[] (4, -3.5) rectangle (5, -6.5);
    \draw[] (5, -3.5) rectangle (7, -6.5);
%    \draw[] (6, -3.5) rectangle (7, -6.5);
    \draw[pattern=north west lines] (7, -3.5) rectangle (8, -6.5);
    \draw[pattern=north west lines] (8, -3.5) rectangle (9, -6.5);
%    \node[] () at (11, -5) {$n-k-t$};
%    \node[] () at (4, -7.1) {$n$};

    \node () at (0.90, -2.5) {$\bm{H}_\text{sub}^{(1)}$};
    \draw[-latex] (0.60,-3.0) -- (0.60,-3.5);
    \node () at (4.90, -2.5) {$\cdots$};
    \node () at (8.90, -2.5) {$\bm{H}_\text{sub}^{(\shots)}$};
    \draw[-latex] (8.60,-3.0) -- (8.60,-3.5);

    \draw[-latex] (2.5,-7.5) -- (2.5,-6.5);
    \draw[-latex] (6.0,-7.5) -- (6.0,-6.5);
%    \draw[-latex] (7.5,-7.5) -- (7.5,-6.5);
    \node[] () at (5.0, -8.0) {all-zero blocks at positions of full-rank errors};
%
%    \node[] (T1) at (16.5, 4.32) {$\in \Fq^{t_i \times n_i}$};
%    \draw[-latex] (T1) -- (T2);
%
    \node[] () at (12, -5) {$\mapsto \bm{H}_\text{sub}^{(i)} = $};
    \node[] () at (12, -6.5) {$\forall i \in [1:\shots]$};

    \draw[] (15, -3.0) rectangle node {$\H_{\text{sub},1}^{(i)}$} (20, -4.5);
    \draw[] (15, -4.5) rectangle node {$\vdots$} (20, -6.0);
    \draw[] (15, -6.0) rectangle node {$\H_{\text{sub},n-k-t}^{(i)}$} (20, -7.5);

\end{tikzpicture}
%    \caption{Illustration of the error support for the sum-rank-metric case. \fh{Why is this caption so short and the others are not? I'm fine with either way but I think we should be consistent.}}\label{fig:sum_rank_figure}
    \caption{Illustration of the error support for the sum-rank-metric case with $\E\in\Fqm^{\intOrder \times n}$, $\A\in\Fqm^{\intOrder\times t_{\Sigma R}}$, $\B\in\Fq^{t_{\Sigma R}\times n}$ and $\Hsub\in\Fqm^{(n-k-t_{\Sigma R})\times n}$. $\B_i$ is the $i$-th row of $\B$ and $\H_{\text{sub},i}$ the $i$-th row of $\Hsub$.}\label{fig:sum_rank_figure}
\end{figure}

%\input{sections/generalizations.tex}
%----------------------------------------------------------------------------
% Conclusion
%----------------------------------------------------------------------------
\section{Conclusion} \label{sec:conclusion}

%\tj{also here maybe a sentence about complexity of code without strucutre.}

We studied the decoding of homogeneous $s$-interleaved sum-rank-metric codes that are obtained by vertically stacking
$s$ codewords of the same arbitrary linear constituent code $\mycode{C}$ over $\Fqm$.
The proposed \MeKap-like decoder for the sum-rank metric relies on linear-algebraic operations only and has a complexity of
$\oh{\max\{n^3, n^2\intOrder\}}$ operations in $\Fqm$, where $n$ denotes the length of $\mycode{C}$.
The decoder works for any linear constituent code and therefore the decoding complexity is not affected by the decoding complexity of the constituent code.
The proposed \MeKap-like decoder can guarantee to correct error matrices $\E \in \Fqm^{\intOrder \times n}$ of sum-rank weight $t \leq d-2$, where $d$ is the minimum distance of $\mycode{C}$, if $\E$ has full $\Fqm$-rank $t$, which implies the high-order condition $s \geq t$.

As the sum-rank metric generalizes both, the Hamming metric and the rank metric, Metzner and Kapturowski's decoder in the
Hamming metric and its analog in the rank metric are both recovered as special cases from our proposal.
Moreover, we showed how the presented algorithm can be used to solve the decoding problem of some high-order interleaved
skew-metric codes.

Since the decoding process is independent of any structural knowledge about the constituent code, this result has a high
impact on the design and the security-level estimation of new code-based cryptosystems in the sum-rank metric.
In fact, if high-order interleaved codes are e.g. used in a classical McEliece-like scheme, any error of sum-rank weight $t \leq d-2$ with full
$\Fqm$-rank $t$ can be decoded without knowledge of the private key.
This directly renders this approach insecure and shows that the consequences of the presented results need to be carefully considered for the design of quantum-resistant public-key systems.

We conclude the paper by giving some further research directions:
The proposed decoder is capable of decoding an error correctly as long as it satisfies the full-rank condition and has
sum-rank weight at most $d-2$, where $d$ denotes the minimum distance of the constituent code.
Similar to Haslach and Vinck's work~\cite{haslach2000efficient} in the Hamming metric, it could be interesting to
abandon the full-rank condition and study a decoder that can also handle linearly dependent errors.
Another approach, that has already been pursued in the Hamming and the rank metric~\cite{oh1994performance,puchinger2019decoding},
is to allow error weights exceeding $d-2$ and investigate probabilistic decoding.

Moreover, an extension of the decoder to heterogeneous interleaved codes (cp.~\cite{puchinger2019decoding} for the
rank-metric case) and the development of a more general decoding framework for high-order interleaved skew-metric codes
can be investigated.

%----------------------------------------------------------------------------
% References
%----------------------------------------------------------------------------
% \section*{Acknowledgment} 
\bibliographystyle{splncs04}
\bibliography{references}

%----------------------------------------------------------------------------
% Appendix
%----------------------------------------------------------------------------
% \clearpage
% \section*{Appendix} 
% \appendices

% \addtolength{\labelsep}{-1pt}
% \vspace{3pt}

\end{document}